\documentclass[submission,copyright,creativecommons]{eptcs}
\providecommand{\event}{GandALF 2020} 
\providecommand{\publicationstatus}{\ifarxivelse{Full version of a paper to appear in GandALF 2020}{Submitted to: GandALF 2020}}

\usepackage{amsmath,amssymb,amsthm,graphicx,wrapfig,paralist}
\usepackage{algorithmicx,algorithm}
\usepackage[noend]{algpseudocode}
\usepackage{booktabs}
\usepackage[table]{xcolor}
\usepackage{pdflscape}
\usepackage{proof}
\usepackage{tikz}
\usepackage{csquotes}
\usepackage{graphics}
\usepackage{stmaryrd}
\usepackage[caption=false]{subfig}

\usepackage{ellipsis, mparhack, ragged2e} 
\usepackage[l2tabu, orthodox]{nag} 

\usepackage{tabu,multirow}
\usepackage{xparse}
\usepackage{mathtools}
\usepackage{environ}

\usepackage[english]{babel}
\usepackage{hyperref}
\usepackage{url}

\usepackage{caption}

\usepackage{algorithm}
\usepackage{algpseudocode}

\usepackage{marginfix}
\usepackage{xifthen}

\usepackage{listings}
\lstdefinestyle{customJ}{
  belowcaptionskip=1\baselineskip,
  breaklines=true,
  frame=L,
  xleftmargin=\parindent,
  language=Java,
  showstringspaces=false,
  basicstyle=\footnotesize\ttfamily,
  keywordstyle=\bfseries\color{green!40!black},
  commentstyle=\itshape\color{purple!40!black},
  identifierstyle=\color{blue},
  stringstyle=\color{orange},
}

\usepackage{afterpage}

\usepackage{xfrac}

\usepackage{rotating}

\usetikzlibrary{shapes,positioning,arrows,calc,automata,matrix,fit}
\tikzstyle{state}+=[minimum size = 6mm, inner sep=0,outer sep=1]
\tikzset{->,>=stealth'}

\algrenewcommand{\algorithmiccomment}[1]{\hskip1.5em \textbackslash *  #1 * \textbackslash}

\newcolumntype{L}{X}
\newcolumntype{R}{>{\raggedleft\arraybackslash}X}
\newcolumntype{C}{>{\centering\arraybackslash}X}

\DeclareGraphicsExtensions{.pdf, .png, .jpg}

\RequirePackage{marvosym}

\DeclarePairedDelimiter{\delimabs}{\lvert}{\rvert}
\DeclarePairedDelimiter{\delimnorm}{\lVert}{\rVert}
\DeclarePairedDelimiter{\delimpospart}{\lgroup}{\rgroup^+}

\DeclarePairedDelimiterX{\deliminner}[2]{\lange}{\rangle}{#1, #2}
\DeclarePairedDelimiter{\delimcardinality}{\lvert}{\rvert}
\DeclarePairedDelimiter{\delimset}{\lbrace}{\rbrace}
\DeclarePairedDelimiter{\delimtuple}{(}{)}
\DeclarePairedDelimiter{\delimlistt}{[}{]}
\DeclarePairedDelimiter{\delimfun}{(}{)}

\NewDocumentCommand{\abs}{sm}{\IfBooleanTF{#1}{\delimabs{#2}}{\delimabs*{#2}}}
\NewDocumentCommand{\norm}{sm}{\IfBooleanTF{#1}{\delimnorm{#2}}{\delimnorm*{#2}}}
\NewDocumentCommand{\pospart}{sm}{\IfBooleanTF{#1}{\delimpospart{#2}}{\delimpospart*{#2}}}
\NewDocumentCommand{\negpart}{sm}{\IfBooleanTF{#1}{\delimnetpart{#2}}{\delimnetpart*{#2}}}
\NewDocumentCommand{\inner}{sm}{\IfBooleanTF{#1}{\deliminner{#2}}{\deliminner*{#2}}}
\NewDocumentCommand{\cardinality}{sm}{\IfBooleanTF{#1}{\delimcardinality{#2}}{\delimcardinality*{#2}}}
\NewDocumentCommand{\set}{sm}{\IfBooleanTF{#1}{\delimset*{#2}}{\delimset{#2}}}
\NewDocumentCommand{\tuple}{sm}{\IfBooleanTF{#1}{\delimtuple{#2}}{\delimtuple*{#2}}}
\NewDocumentCommand{\closure}{sm}{\IfBooleanTF{#1}{\delimclosure{#2}}{\delimclosure*{#2}}}
\NewDocumentCommand{\listt}{sm}{\IfBooleanTF{#1}{\delimlistt{#2}}{\delimlistt*{#2}}}
\NewDocumentCommand{\fun}{smm}{\IfBooleanTF{#1}{{#2}\delimfun{#3}}{{#2}\delimfun*{#3}}}
\NewDocumentCommand{\funMacro}{smm}{\IfNoValueTF{#3}{#1}{\fun{#2}{#3}}}

\DeclareMathOperator{\ExistsOp}{\exists}
\DeclareMathOperator{\ForallOp}{\forall}

\NewDocumentCommand{\Exists}{gg}{\IfNoValueTF{#1}{\ExistsOp}{\ExistsOp #1. \, #2}}
\NewDocumentCommand{\Forall}{gg}{\IfNoValueTF{#1}{\ForallOp}{\ForallOp #1. \, #2}}

\newcommand{\intersectionSym}{\cap}
\newcommand{\intersectionBin}{\mathbin{\intersectionSym}}
\newcommand{\UnionSym}{\bigcup}

\newcommand{\intersection}{\intersectionBin}
\newcommand{\Union}{\UnionSym}

\newcommand{\Naturals}{\mathbb{N}}

\newcommand{\Distributions}{\mathcal{D}}

\NewDocumentCommand{\convto}{G{}}{\xrightarrow{#1}}
\NewDocumentCommand{\weakto}{G{}}{\xrightharpoonup{#1}}
\NewDocumentCommand{\weakstarto}{G{}}{\xrightharpoonup[*]{#1}}

 \DeclareDocumentCommand{\diff}{D<>{} O{}  D(){}}{\Delta_{#1}^{#2}\ifthenelse{\isempty{#3}}{}{(#3)}}

\DeclareMathOperator*{\argmax}{arg\, max}
\DeclareMathOperator*{\argmin}{arg\, min}
\newcommand{\post}{\mathsf{Post}}

\newcommand{\eqdef}{\vcentcolon=}

\newcommand{\qee}{\hfill$\triangle$} 

\newcommand{\np}{\mathbf{NP}}
\newcommand{\conp}{\mathbf{co\text{-}NP}}

\NewDocumentCommand{\distributions}{d()}{\funMacro{\mathcal{D}}{#1}}

\DeclareDocumentCommand{\val}{D<>{} O{}  D(){} t'}{\mathsf{V}_{#1}^{\IfBooleanTF{#4}{\prime #2}{#2}}\ifthenelse{\isempty{#3}}{}{(#3)}}
\DeclareDocumentCommand{\ub}{D<>{} O{}  D(){} t'}{\mathsf{U}_{#1}^{\IfBooleanTF{#4}{\prime #2}{#2}}\ifthenelse{\isempty{#3}}{}{(#3)}}
\DeclareDocumentCommand{\gub}{D<>{} O{}  D(){} t'}{\mathsf{G}_{#1}^{\IfBooleanTF{#4}{\prime #2}{#2}}\ifthenelse{\isempty{#3}}{}{(#3)}}
\DeclareDocumentCommand{\lb}{D<>{} O{}  D(){} t'}{\mathsf{L}_{#1}^{\IfBooleanTF{#4}{\prime #2}{#2}}\ifthenelse{\isempty{#3}}{}{(#3)}}
\DeclareDocumentCommand{\game}{D<>{} O{} D(){} t'}{\mathsf{G}_{#1}^{\IfBooleanTF{#4}{\prime}{}#2}\ifthenelse{\isempty{#3}}{}{(#3)}}
\DeclareDocumentCommand{\transition}{D<>{} O{} D(){}}{\rightarrow_{#1}^{#2}\ifthenelse{\isempty{#3}}{}{(#3)}}

\newcommand{\SG}{\textrm{SG}}
\newcommand{\SGs}{\textrm{SGs}}
\DeclareDocumentCommand{\G}{D<>{} O{} t' D(){}}{\mathsf{G}_{#1}^{\IfBooleanTF{#3}{\prime}{}#2}\ifthenelse{\isempty{#4}}{}{(#4)}}
\DeclareDocumentCommand{\exGame}{D<>{} O{} t'  D(){}}{\mathsf{G}_{#1}^{\IfBooleanTF{#3}{\prime}{}#2}\ifthenelse{\isempty{#4}}{}{(#4)}=(\states<#1>[\IfBooleanTF{#3}{\prime}{}#2],\states<\Box\ifthenelse{\isempty{#1}}{}{,#1}>[\IfBooleanTF{#3}{\prime}{}#2],\states<\circ\ifthenelse{\isempty{#1}}{}{,#1}>[\IfBooleanTF{#3}{\prime}{}#2],\istate<#1>[\IfBooleanTF{#3}{\prime}{}#2],\actions<#1>[\IfBooleanTF{#3}{\prime}{}#2],\Av<#1>[\IfBooleanTF{#3}{\prime}{}#2],\trans<#1>[\IfBooleanTF{#3}{\prime}{}#2])}

\DeclareDocumentCommand{\states}{D<>{} O{}  t'}{\mathit{S}_{#1}^{\IfBooleanTF{#3}{\prime~#2}{#2}}}
\DeclareDocumentCommand{\state}{D<>{} O{}  t'}{\mathsf{s}_{#1}^{\IfBooleanTF{#3}{\prime #2}{#2}}}
\DeclareDocumentCommand{\istate}{D<>{} O{} t'}{\mathsf{s}_{0\ifthenelse{\isempty{#1}}{}{,#1}}^{\IfBooleanTF{#3}{\prime~#2}{#2}}}

\newcommand{\initstate}{\state<0>}

\DeclareDocumentCommand{\trans}{D<>{} O{} t' D(){} D(){}}{\delta_{#1}^{\IfBooleanTF{#3}{\prime}{}#2}\ifthenelse{\isempty{#4}}{}{(#4)}\ifthenelse{\isempty{#5}}{}{(#5)}}
\DeclareDocumentCommand{\Av}{D<>{} O{} t' D(){}}{\mathsf{Av}_{#1}^{\IfBooleanTF{#3}{\prime}{}#2}\ifthenelse{\isempty{#4}}{}{(#4)}}

\DeclareDocumentCommand{\F}{D<>{} O{} t' D(){}}{\mathsf{F}_{#1}^{\IfBooleanTF{#3}{\prime}{}#2}\ifthenelse{\isempty{#4}}{}{(#4)}}

\newcommand{\Path}{\rho}
\DeclareDocumentCommand{\Path}{D<>{} O{} t' D(){}}{\path<#1>[#2]\IfBooleanTF{#3}{'}{}(#4)}
\DeclareDocumentCommand{\path}{D<>{} O{} t' D(){}}{\rho_{#1}^{\IfBooleanTF{#3}{\prime}{}#2}\ifthenelse{\isempty{#4}}{}{(#4)}}
\newcommand{\fpath}{\mathsf{w}}
\DeclareDocumentCommand{\Paths}{D<>{} O{} t' D(){}}{\Omega_{#1}^{\IfBooleanTF{#3}{\prime}{}#2}\ifthenelse{\isempty{#4}}{}{(#4)}}
\newcommand{\straa}{\sigma}

\newcommand{\strab}{\tau}

\DeclareDocumentCommand{\strategy}{D<>{} O{} D(){}
  t*}{{\IfBooleanTF{#4}{\tau}{\sigma}}_{#1}^{#2}\ifthenelse{\isempty{#3}}{}{(#3)}}

\DeclareDocumentCommand{\actions}{D<>{} O{} t' d()}{{\IfNoValueTF{#4}{\mathsf{A}}{\fun{\mathsf{A}}{#4}}}_{#1}^{\IfBooleanTF{#3}{\prime~#2}{#2}}}
\DeclareDocumentCommand{\action}{D<>{} O{} t'}{\mathsf{a}_{#1}^{\IfBooleanTF{#3}{\prime#2}{#2}}}

\newcommand{\mec}{\mathsf{MEC}}

\newcommand{\pr}{\mathbb P}

\DeclareDocumentCommand{\target}{D<>{} O{}
  t'}{\mathfrak{1}_{#1}^{\IfBooleanTF{#3}{\prime}{}#2}}
\DeclareDocumentCommand{\fail}{D<>{} O{} t'}{\mathsf{\bot}_{#1}^{\IfBooleanTF{#3}{\prime}{}#2}}

\DeclareDocumentCommand{\CEC}{D<>{} O{} D(){}
  t*}{\IfBooleanTF{#4}{\simple~}{}\ifthenelse{\isempty{#3}}{}{#3\textit{-}}\textrm{BEC}_{#1}^{#2}}

\newcommand{\simple}{simple}

\NewDocumentCommand{\bE}{D<>{}D[]{}D(){}}{\mathsf{best\_exit}_{#1}^{#2}\ifthenelse{\isempty{#3}}{}{(#3)}}

\newcommand{\FIND}{{\textsf{FIND\_SEC}}}

\newcommand{\fstates}{\mathit{F}}

\newcommand{\drawcirc}{\node[draw,circle,minimum size=.7cm, outer sep=1pt]}
\newcommand{\drawbox}{\node[draw,rectangle,minimum size=.7cm, outer sep=1pt]}
\newcommand{\drawdummy}{\node[minimum size=0,inner sep=0]}

\newcommand{\sink}{\mathit{Z}}

\newcommand{\cnf}{Condon's normal form}

\newcommand{\twoSucc}{\mathrm{\textbf{2Act}}}

\newcommand{\halfProbs}{\mathrm{\pmb{\sfrac{1}{2}}\textbf{Probs}}}
\newcommand{\stopping}{\mathrm{\textbf{Stopping game}}}
\newcommand{\noOneSucc}{\mathrm{\textbf{No1Act}}}
\DeclareDocumentCommand{\actionb}{D<>{} O{} t'}{\mathsf{b}_{#1}^{\IfBooleanTF{#3}{\prime#2}{#2}}}

\newcommand{\zeroSink}{\mathrm{o}}
\newcommand{\SGinstance}{\mathbf{G}}
\newcommand{\exits}{E^T}
\newcommand{\valstra}{\mathsf{\val}_{\straa ,\strab}}

\definecolor{coolgrey}{rgb}{0.55, 0.57, 0.67}
\renewcommand{\algorithmiccomment}[1]{\textcolor{coolgrey!80!black}{\# #1}}

\renewcommand{\circ}{\bigcirc}
\newtheorem{definition}{Definition}
\newtheorem{example}{Example}

\newtheorem{lemma}{Lemma}

\usepackage{etoolbox} 
\newtoggle{isArxiv}
\newcommand{\ifarxivelse}[2]{\iftoggle{isArxiv}{#1}{#2}}
\toggletrue{isArxiv} 

\title{Comparison of Algorithms for Simple Stochastic Games\ifarxivelse{ \\(Full Version)}{}
	\thanks{This research was funded in part by the German Research Foundation (DFG) projects 383882557  \emph{Statistical Unbounded Verification (SUV)} and 427755713 \emph{Group-By Objectives in Probabilistic Verification (GOPro)}.}}
	\author{Jan K\v ret\'insk\'y, Emanuel Ramneantu, Alexander Slivinskiy, Maximilian Weininger \institute{Technical University of Munich} \email{\{jan.kretinsky,emanuel.ramneantu,alexander.slivinskiy,maxi.weininger\}@tum.de}}
\def\titlerunning{Comparison of Algorithms for Simple Stochastic Games}
\def\authorrunning{J. K\v ret\'insk\'y, E. Ramneantu, A. Slivinskiy and M. Weininger}

\begin{document}
	\maketitle

\begin{abstract}	
%

	Simple stochastic games are turn-based 2\textonehalf-player zero-sum graph games with a reachability objective. 
	The problem 
	is to compute the winning probability as well as the optimal strategies of both players. 
	In this paper, we compare the three known classes of algorithms -- value iteration, strategy iteration and quadratic programming -- both theoretically and practically.
	Further, we suggest several improvements for all algorithms, including the first approach based on quadratic programming that avoids transforming the stochastic game to a stopping one. 
	Our extensive experiments show that these improvements can lead to significant speed-ups.
	We implemented all algorithms in PRISM-games 3.0, thereby providing the first implementation of quadratic programming for solving simple stochastic games.
	
\end{abstract} 

\maketitle

\section{Introduction} \label{sec:intro}

\emph{Simple stochastic games} (SGs), e.g. \cite{condonComplexity}, are zero-sum 
games played on a graph by 
players Maximizer and Minimizer, who choose actions in their respective vertices (also called states).
Each action is associated with a probability distribution determining the next state to move to.
The objective of Maximizer is to maximize the probability of reaching a given target state; the objective of Minimizer is the opposite.

The basic decision problem is to determine whether Maximizer can ensure a reachability probability above a certain threshold if both players play optimally.
This problem is among the rare and intriguing combinatorial problems that are in $\np \cap \conp$~\cite{DBLP:conf/dimacs/Condon90}, but 
whether it belongs to $\mathbf{P}$ is a major and long-standing open problem.
Further, several other important problems can be reduced to SG, for instance parity games, mean-payoff games, discounted-payoff games and their stochastic extensions~\cite{DBLP:journals/corr/abs-1106-1232}.

Besides the theoretical interest, SGs are a standard model in control and verification of stochastic 
reactive systems~\cite{FV97,CH12}; see e.g. \cite{DBLP:journals/ejcon/SvorenovaK16} for an overview over various recent case studies.
Further, since Markov decision processes (MDP) \cite{Puterman} are a special case with only one player, SGs can serve as abstractions of large MDPs~\cite{DBLP:journals/fmsd/KattenbeltKNP10} or provide robust versions of MDPs when precise transition probabilities are not known~\cite{DBLP:conf/fossacs/ChatterjeeSH08,WMK19}.

There are three classes of algorithms for computing the reachability probability in simple stochastic games, as surveyed in~\cite{DBLP:conf/dimacs/Condon90}:
value iteration (VI), strategy iteration (SI, also known as policy iteration) and quadratic programming (QP).
In~\cite{DBLP:conf/dimacs/Condon90}, they all required the SG to be transformed into a certain normal form, among other properties 
ensuring that the game is stopping. This not only blows up the size of the SG, but also changes the reachability probability; however, it is possible to infer the original probability.
For VI and SI, this requirement has since been lifted, e.g.~\cite{visurvey,CAH13}, but not for QP.

While searching for a polynomial algorithm, there were several papers on VI and SI; however, the theoretical improvements so far are limited to subexponential variants~\cite{DBLP:journals/iandc/Ludwig95,DBLP:journals/algorithmica/DaiG11} and variants that are fast for SG with few random vertices~\cite{GH08,IM12}, i.e. games where most actions yield a successor deterministically.
QP was not looked at, as it is considered common knowledge that it performs badly in practice.

There exist several tools for solving games:
GAVS+~\cite{DBLP:conf/tacas/ChengKLB11} offers VI and SI for SGs, among other things. However, it is more for educational purposes than large case studies and currently not maintained.
GIST~\cite{DBLP:conf/cav/ChatterjeeHJR10} performs qualitative analysis of stochastic games with $\omega$-regular objectives.
For MDPs (games with a single player), we refer to~\cite{DBLP:conf/tacas/HahnHHKKKPQRS19} for an overview of existing tools.
Most importantly, PRISM-games~3.0~\cite{PRISM-games3} is a recent tool that offers algorithms for several classes of games. 
However,
for SGs with a reachability objective, it offers only variants of value iteration, none of which give a guarantee, so the result might be arbitrarily far off.
Thus, currently no tool offers a precise solution method for large simple stochastic games.

\textbf{Our contribution} is the following:
\begin{itemize}
	\item We provide an extension of the quadratic programming approach that does not require the transformation of the SG into a stopping SG in normal form.
	\item We propose several optimizations for the algorithms, inspired by~\cite{KM17}, including an extension of topological value iteration from MDPs~\cite{TVI1, ensure}.
	\item We implement 
	VI with guarantees on precision as well as SI, QP and our optimizations in PRISM-games 3.0---thereby providing the first implementations of QP for SGs---and experimentally compare them on both the realistic case studies of PRISM-games and on interesting handcrafted corner cases.
\end{itemize}

\subsection*{Related Work}

We sketch the recent developments of each class of algorithms 
and then several
further directions. 

Value iteration is a standard solution method also for MDPs~\cite{Puterman}. For a long time, the stopping criterion of VI was such that it could return arbitrarily imprecise results~\cite{BVI}. In principle, the computation has to run for an exponential number of steps~\cite{visurvey}. 
However, recently several heuristics that give guarantees and are typically fast were proposed for MDPs~\cite{atva,BVI,DBLP:conf/cav/QuatmannK18,OVI}, as well as for SGs~\cite{KKKW18}. 
A learning-based variant of VI for MDPs~\cite{atva} was also extended to SGs~\cite{KKKW18}.
The variant of VI from~\cite{IM12} requires the game to be in normal form, blowing up its size, and is good only if there are few random vertices, as it depends on the factorial of this number; it is impractical for almost all case studies considered in this paper.

Strategy iteration was introduced in~\cite{HK66}.
A randomized version was given in~\cite{DBLP:conf/dimacs/Condon90}, and subexponential versions in~\cite{DBLP:journals/iandc/Ludwig95,DBLP:journals/algorithmica/DaiG11}.
Another variant of SI was proposed in~\cite{DBLP:journals/entcs/Somla05}, however there is no evidence that it runs in polynomial time.
The idea of considering games with few random vertices (as in the already mentioned works of \cite{DBLP:journals/algorithmica/DaiG11,IM12}) was first introduced in~\cite{GH08}, where they exhaustively search a subspace of strategies, which is called f-strategies.

To the best of our knowledge, quadratic programming as solution method for simple stochastic games was not investigated further after the first mention in \cite{DBLP:conf/dimacs/Condon90}. 
However, convex QPs are solvable in polynomial time~\cite{kozlov1980polynomial}.
If it was possible to encode the problem in a convex QP of polynomial size, this would result in a polynomial algorithm.
The encoding we provide in this paper can however be exponential in the size of the game.

Further related works consider other variants of the model, for example concurrent stochastic games, see e.g. \cite{DBLP:journals/mst/HansenIM14} for the complexity of SI and VI and \cite{DBLP:conf/mfcs/ChatterjeeHI17} for strategy complexity, or games with limited information~\cite{DBLP:conf/cav/AshokKW19}.
Furthermore, one can consider other objectives, e.g. $\omega$-regular objectives~\cite{CH12}, mean payoff~\cite{DBLP:journals/tcs/ZwickP96} or combinations of objectives~\cite{DBLP:conf/mfcs/ChenFKSW13,DBLP:conf/lics/AshokCKWW20}.

\section{Preliminaries} \label{sec:prelim}
We now introduce the model of stochastic games, define the semantics by the standard means of infinite paths and strategies and then define the important concept of end components, which are subgraphs of stochastic games that are problematic for all three classes of algorithms.

A \emph{probability distribution} on a finite set $X$ is a mapping $\trans: X \to [0,1]$, such that $\sum_{x\in X} \trans(x) = 1$.
The set of all probability distributions on $X$ is denoted by $\Distributions(X)$.

\subsection{Stochastic Games}

%
%
Now we define stochastic games, in literature often referred to as simple stochastic games or turn-based stochastic two-player games with a reachability objective.
As opposed to the notation of e.g.\ \cite{condonComplexity}, we do not have special stochastic nodes, but rather a probabilistic transition function.

\begin{definition}[\SG]
	A \emph{stochastic game ($\SG$)} is a tuple 
	$(\states,\states<\Box>,\states<\circ>,\initstate,\actions,\Av,\delta)$ 
	where 
	\begin{itemize}
		\item 	$\states$ is a finite set of \emph{states} partitioned\footnote{I.e., $\states<\Box> \subseteq \states$, $\states<\circ> \subseteq \states$, $\states<\Box> \cup \states<\circ> = \states$, and $\states<\Box> \cap \states<\circ> = \emptyset$.}\ into the sets $\states<\Box>$ and $\states<\circ>$ of states of the player \emph{Maximizer} and \emph{Minimizer}\footnote{The names are chosen, because Maximizer maximizes the probability of reaching the given target states, and Minimizer minimizes it.} respectively
		\item $\initstate \in \states$ is the \emph{initial} state
		\item $\actions$ is a finite set of \emph{actions}
		\item $\Av: \states \to 2^{\actions}$ assigns to every state a set of \emph{available} actions
		\item $\trans: \states \times \actions \to \distributions(\states)$ is a \emph{transition function} that given a state $\state$ and an action $\action\in \Av(\state)$ yields a probability distribution over \emph{successor} states.
		We slightly abuse notation and write $\trans(\state,\action,\state')$ instead of $\trans(\state,\action)(\state')$.
	\end{itemize}
	See Figure \ref{fig:SGex} for an example of an SG.
	Since we consider the reachability objective, the $\SG$ is complemented with a set of target states $\fstates \subseteq \states$.
	What happens after reaching any target state is irrelevant for the reachability probability, so we can assume that every target state only has one action that is a self-loop with probability~1.
	A \emph{Markov decision process (MDP)} is a special case of $\SG$ where $\states<\circ> = \emptyset$ , and a Markov chain (MC) is a special case of an MDP, where for all $\state \in \states: \abs{\Av(\state)} = 1$. 
\end{definition}
Without loss of generality we assume that $\SGs$ are non-blocking, so for all states $\state$ we have $\Av(\state) \neq \emptyset$.
For a state $\state$ and an available action $\action \in \Av(\state)$, we denote the set of successors by $\post(\state,\action) := \set{\state' \mid \trans(\state,\action,\state') > 0}$.
Finally, for any set of states $T \subseteq \states$, we use $T_\Box$ and $T_\circ$ to denote the states in $T$ that belong to Maximizer and Minimizer, whose states are drawn in the figures as $\Box$ and $\circ$, respectively.


\begin{figure}[t]

\vspace{-1em}
\centering
\begin{tikzpicture}
\drawdummy (init) at (0,0) {};
\drawcirc (p) at (1,0) {$\mathsf{p}$};
\drawbox (q) at (3,0) {$\mathsf{q}$};
\drawdummy (mid) at (4,0) {};
\drawbox (1) at (6,0.5) {$\target$ };
\drawcirc (0) at (6,-0.5)  {$\mathrm{\zeroSink}$};

\draw[->] (init) to (p);
\draw[->]  (p) to[bend left] node [midway,anchor=south] {$\mathsf{a}$}(q) ;
\draw[->]  (q) to [bend left] node [midway,anchor=north] {$\mathsf{b}$} (p);
\draw[-] (q) to node [midway,anchor=south] {$\mathsf{c}$} (mid) ;
\draw[->] (mid) to node [below] {$\sfrac13$} (0);
\draw[->] (mid) to node [above] {$\sfrac13$} (1);
\draw[->] (mid) to [bend left,out=270,in=270] node[above] {$\sfrac13$} (q) ;
\draw[->]  (0) to[loop right]  node [midway,anchor=west] {$\mathsf{e}$} (0);
\draw[->]  (1) to [loop right] node [midway,anchor=west] {$\mathsf{d}$} (1) ;

\end{tikzpicture}
\caption{An example of an $\SG$ with $\states = \set{\mathsf{p},\mathsf{q},\target,\mathrm{\zeroSink}}$, $\states<\Box> = \set{\mathsf{q},\target}$, $\states<\circ> = \set{\mathsf{p},\mathrm{\zeroSink}}$, the initial state $\mathsf{p}$ and the set of actions $\actions = \set{\mathsf{a},\mathsf{b},\mathsf{c},\mathsf{d},\mathsf{e}}$; $\Av(\mathsf{p})=\set{\mathsf{a}}$ with $\trans(\mathsf{p},\mathsf{a})(\mathsf{q})=1$;  $\Av(\mathsf{q})=\set{\mathsf{b},\mathsf{c}}$ with $\trans(\mathsf{q},\mathsf{b})(\mathsf{p})=1$ and $\trans(\mathsf{q},\mathsf{c})(\mathsf{q}) = \trans(\mathsf{q},\mathsf{c})(\target) = \trans(\mathsf{q},\mathsf{c})(\mathrm{\zeroSink}) = \frac{1}{3}$. 
	For actions with only one successor, we do not depict the transition probability $1$.
}
\label{fig:SGex}
\vspace{-0.5em}
\end{figure}

\subsection{Semantics: Paths, Strategies and Values}\label{sec:semantics}
The semantics of SGs is given in the usual way by means of strategies, the induced Markov chain and the respective probability space, as follows:
An \emph{infinite path} $\path$ is an infinite sequence $\path = \state<0> \action<0> \state<1> \action<1> \dots \in (\states \times \actions)^\omega$, such that for every $i \in \Naturals$ we have $\action<i>\in \Av(\state<i>)$ and $\state<i+1> \in \post(\state<i>,\action<i>)$.
\emph{Finite path}s are defined analogously as elements of $(\states \times \actions)^\ast \times \states$.

As this paper deals with the reachability objective, we can restrict our attention to memoryless deterministic strategies, which are optimal for this objective~\cite{condonComplexity}.
A \emph{strategy} of Maximizer, respectively Minimizer, is a function $\straa: \states<\Box> \to \actions$, respectively $\states<\circ> \to \actions$, such that $\straa(\state) \in \Av(\state)$ for all $\state$.
A pair $(\straa,\strab)$ of memoryless deterministic strategies of Maximizer and Minimizer induces a Markov chain $\G[\straa,\strab]$, as for every state there is only one available action.
The Markov chain 
induces a unique probability distribution $\pr_{\state}^{\straa,\strab}$ over measurable sets of infinite paths \cite[Ch.~10]{BK08}. 

We write $\Diamond \fstates:=\set{\path \mid \path = \state<0> \action<0> \state<1> \action<1> \dots \in (\states \times \actions)^\omega \wedge \exists i \in \Naturals.~\state<i> \in \fstates}$ to denote the (measurable) set of all paths which eventually reach $\fstates$. 
For each $\state\in\states$, we define the \emph{value} in $\state$ as 
\[\val(\state) \eqdef \sup_{\straa} \inf_{\strab} \pr_{\state}^{\straa,\strab}(\Diamond \fstates)= \inf_{\strab} \sup_{\straa}\pr_{s}^{\straa,\strab}(\Diamond \fstates),\]

\noindent where the equality follows from~\cite{condonComplexity}.
The value is the least fixpoint of the so called \emph{Bellman equations}~\cite{DBLP:conf/dimacs/Condon90}:
\begin{equation}\label{eq:Vs}
\val(\state) =  \begin{cases} 
1 &\mbox{if } \state \in \fstates \\
\max_{\action \in \Av(\state)}\val(\state,\action)		&\mbox{if } \state \in \states<\Box> \setminus \fstates  \\
\min_{\action \in \Av(\state)}\val(\state,\action) &\mbox{if } \state \in \states<\circ> \setminus \fstates
\end{cases}\text{,}
\end{equation} 
\begin{eqnarray}\label{eq:Vsa}
\text{with~} \val(\state,\action) \eqdef \sum_{s' \in S} \trans(\state,\action,\state') \cdot \val(\state')
\end{eqnarray}

The states with no path to the target, so called sinks, are of special interest. We denote the set of sinks as $\sink$. Sinks have a value of 0 and can be found a priori by graph analysis.

We are interested not only in the values $\val(\state)$ for all $\state \in \states$, but also their $\varepsilon$-approximation, i.e. an approximation $\lb$ with $\abs{\val(\state) - \lb(\state)} < \varepsilon$; as well as the corresponding ($\varepsilon$-)optimal strategies for both players, i.e. a pair of strategies $(\straa,\strab)$ under which $\pr_{\state}^{\straa,\strab}(\Diamond \fstates)$ is equal to $\val(\state)$ (respectively an $\varepsilon$-approximation of $\val(\state)$).
Note that it suffices to have either the values or the optimal strategies, because from one we can infer the other.
Given a pair of optimal strategies $(\straa,\strab)$, the values can be computed by solving the induced Markov chain $\G[\straa,\strab]$.
Given a vector of values for all states, an optimal pair of strategies can be computed by randomizing over all locally optimal actions in each state (e.g. $\straa(\state)$ randomizes over the set $\argmax_{\action \in \Av(\state)} \val(\state,\action)$ for Maximizer states, and dually with $\argmin$ for Minimizer). To get a deterministic strategy, we cannot only pick some locally optimal action, but additionally we have to ensure that Maximizer is not stuck in some cycle when playing the actions. This works by looking at end components, which are the topic of the next subsection.

\subsection{End Components}
When computing the values of states in an SG, we need to take special care of \emph{end components} (EC).
Intuitively, an EC is a subset of states of the SG, where the game can remain forever;
i.e. given certain strategies of both players, there is no positive probability to exit the EC to some other state.
ECs correspond to bottom strongly connected components of the Markov chains induced by some pair of strategies.

\begin{definition}[EC]
	\label{def:EC}
	A non-empty set $T\subseteq \states$ of states is an \emph{end component (EC)} if there exists a non-empty set $B \subseteq \Union_{\state \in T} \Av(s)$ of actions\footnote{Note that this assumes that action names are unique. This can always be achieved by renaming actions, e.g. prepending every action with the state it is played from} such that 
	\begin{enumerate}
		\item for each $\state \in T, \action \in B \intersection \Av(\state)$ we have $\post(\state,\action) \subseteq T$,
		\item for each $\state, \state' \in T$ there is a finite path $\fpath = \state \action<0> \dots \action<n> \state' \in (T \times B)^* \times T$, i.e. the path stays inside $T$ and only uses actions in $B$.
	\end{enumerate}
An end component $T$ is a \emph{maximal end component (MEC)} if there is no other end component $T'$ such that $T \subseteq T'$.
\end{definition}
Given an $\SG$ $\G$, the set of its MECs is denoted by $\mec(\G)$ and can be computed in polynomial time~\cite{CY95}.
\begin{example}
	Consider the SG of Figure \ref{fig:SGex}. The set of states $T = \{\mathsf p,\mathsf q\}$ is an EC, as when playing only actions from $B = \{\mathsf a, \mathsf b\}$ the play remains in $T$ forever. It is even a MEC, as there is no superset of $T$ with this property. \qee
\end{example}

ECs are of special interest, because in ECs there can be multiple fixpoints of the Bellman equations (see Equation \ref{eq:Vs} and \ref{eq:Vsa}). 
Thus, methods relying on iterative over-approximation of the value do not converge, as they are stuck at some greater fixpoint than the value (cf. \cite[Lemma 1]{KKKW18} and Example \ref{ex:ubNoConvergeVI} in the next section).
This is why the original description of the algorithms~\cite{DBLP:conf/dimacs/Condon90} considered only stopping games, i.e. games without ECs, except for a dedicated target and sink state.
The algorithms are theoretically applicable to arbitrary SGs, as for every non-stopping SG one can construct a stopping SG and infer the original value from solving the stopping SG. See \ifarxivelse{Appendix~\ref{app:stopping}}{\cite[Appendix A.4]{techreport}} for a description of this approach and Section~\ref{sec:QP_stopping} for a discussion of the practical drawbacks. 

\section{State of the Art Algorithms}

In this section, we describe the existing algorithms for solving simple stochastic games.
All of them require as input an SG $\game = (\states,\states<\Box>,\states<\circ>,\initstate,\actions,\Av,\delta)$ and a target set $\fstates$. Value iteration additionally needs a precision~ $\varepsilon$.
After termination, all of them return a vector of values for each state ($\varepsilon$-precise for BVI) and the corresponding ($\varepsilon$-)optimal strategies. 
Quadratic programming in its current form only works on stopping SGs in a certain normal form.

\subsection{Bounded Value Iteration}\label{sec:bvi}
\emph{Value Iteration (VI)}, see e.g.~\cite{Puterman}, is the most common algorithm for solving MDPs and SGs, and the only method implemented in PRISM-games~\cite{PRISM-games3}.
Originally, VI only computed a convergent sequence of under-approximations; however, as it was unclear how to stop, results returned by model checkers could be off by arbitrary amounts~\cite{BVI}.
Thus, it was extended to also compute a convergent over-approximation~\cite{KKKW18}. The resulting algorithm is called \emph{bounded value iteration (BVI)}.

The basic idea of BVI is to start from a vector $\lb<0>$ respectively $\ub<0>$ that definitely is an under-/over-approximation of the value, i.e. for every state $\lb<0>(\state) \leq \val(\state) \leq \ub<0>(\state)$.
Then the algorithm repeatedly applies so called \emph{Bellman updates}, i.e. it uses a version of Equation \ref{eq:Vs} as follows (the equation for the over-approximation is obtained by replacing $\lb$ with $\ub$):
\begin{equation}\label{eq:Ls}
\lb<n>(\state) =  \begin{cases} 
\max_{\action \in \Av(\state)}\lb<n-1>(\state,\action)		&\mbox{if } \state \in \states<\Box>  \\
\min_{\action \in \Av(\state)}\lb<n-1>(\state,\action) &\mbox{if } \state \in \states<\circ>
\end{cases},
\end{equation}
where $\lb<n-1>(\state,\action)$ is computed from $\lb<n-1>(\state)$ as in Equation \ref{eq:Vsa}.
Since $\val$ is the least fixpoint of the Bellman equations, $\lim_{n \to \infty} \lb<n>$ converges to $\val$.
However, the over-approximation $\ub$ need not converge in the presence of ECs.

\begin{example}\label{ex:ubNoConvergeVI}
	Consider the SG of Figure \ref{fig:SGex} with the EC $\{\mathsf p, \mathsf q\}$. Let $\ub<0> = 1$ for $\mathsf p$, $\mathsf q$ and $\target$ and $\ub<0> = 0$ for $\mathrm{\zeroSink}$.
	Then we have $\ub<0>(\mathsf q, \mathsf b) = 1$ and $\ub<0>(\mathsf q, \mathsf c) = \sfrac 2 3$. Thus, $\mathsf q$ will pick action $\mathsf b$, as it promises a higher value, and the over-approximation does not change.
	This happens, because looking at the current upper bound, Maximizer is under the impression that staying in the EC yields a higher value. However, it actually reduces the probability to reach the target to 0.
	So the algorithm has to perform an additional step to inform states in an EC that they should not depend on each other, but on the best exit.
	In this case, $\ub<1>(\mathsf q) = \ub<0>(\mathsf q, \mathsf c) = \sfrac 2 3$. \qee
\end{example}

In fact, it does not suffice to only decrease the upper bound of ECs, but a more in-depth graph analysis is required to detect the problematic subsets of states, so called \emph{simple end components}~(SEC)~\cite{KKKW18}.
Decreasing the value of those SECs to their best exit repeatedly results in a convergent over-approximation. 
\vspace{-1em}
\[\displaystyle \bE<\ub>(T) = \max_{\substack{\state \in T_\Box\\ \neg \post(\state,\action) \subseteq T}} \ub(\state,\action)\] 
is the best exit according to the current estimates of the upper bound.
Algorithm \ref{alg:bvi} shows the full BVI algorithm from \cite{KKKW18}.

\begin{algorithm}[htbp]
	
	\caption{Bounded value iteration algorithm from \cite{KKKW18}}\label{alg:bvi}
	\begin{algorithmic}[1]
		\Procedure{BVI}{precision $\varepsilon>0$}
		\For {$\state \in \states$}\hspace{2.05em}\Comment{Initialization}
		\State $\lb(\state)=0$ \hspace{1.65em}\Comment{Lower bound}
		\State $\ub(\state)=1$ \hspace{1.5em}\Comment{Upper bound}
		\EndFor
		\State \textbf{for} $\state \in \fstates$ \textbf{do} $\lb(\state) = 1$ \hspace{1.5em}\Comment{Value of targets is determined a priori}
		
		\medskip
		
		\Repeat
		\State $\lb,\ub$ get updated according to Eq.~(\ref{eq:Ls})\hspace{1.5em} \Comment{Bellman updates}
		\smallskip
		\For{$T \in \FIND$}\hspace{5.1em}\Comment{For every SEC}
		\For {$\state \in T$} 
		\State $\displaystyle \ub(\state) \gets  \bE<\ub>(T)$\label{line:adjust}\hspace{1.75em} \Comment{Decrease upper bound to best exit}
		\EndFor
		\EndFor
		\Until{$\ub(\state) - \lb(\state) < \varepsilon$ for all $\state \in \states$}\hspace{1.5em}
		\Comment{Guaranteed error bound}
		\EndProcedure
	\end{algorithmic}
\end{algorithm}

There also is a simulation based asynchronous version of BVI (see \cite[Section 4.4]{KKKW18}) which can perform very well on models with a certain structure~\cite{cores}.
It updates states encountered by simulations, and guides those simulations to explore only the relevant part of the state space. If only few states are relevant for convergence, this algorithm is fast; if large parts of the state space are relevant or there are many cycles in the game graph that slow the simulations down, this adaption of BVI is slow.

\subsection{Strategy Iteration}

In contrast to value iteration, the approach of \emph{strategy iteration (SI)}~\cite{HK66} does not compute a sequence of value-vectors, but instead a sequence of strategies.
Starting from an arbitrary strategy of Maximizer, we repeatedly compute the best response of Minimizer and then greedily improve Maximizer's strategy. The resulting sequence of Maximizer strategies is monotonic and converges to the optimal strategy~\cite[Theorem 3]{CAH13}.
The pseudocode for strategy iteration is given in Algorithm \ref{alg:si}.

Note that in non-stopping SGs (games with ECs) the initial Maximizer strategy cannot be completely arbitrary, but it has to be \emph{proper}, i.e. ensure that either a target or a sink state is reached almost surely; it must not stay in some EC, as otherwise the algorithm might not converge to the optimum due to problems similar to those described in Example \ref{ex:ubNoConvergeVI}.
In Algorithm \ref{alg:si}, we use the construction of the \emph{attractor strategy}~\cite[Section 5.3]{CAH13} to ensure that our initial guess is a proper strategy (Line \ref{line:si_attractor}). It first analyses the game graph to find the sink states $\sink$.
Then, it performs a backwards breadth first search, starting from the set of both target and sink states. 
A state discovered in the $i$-th iteration of the search has to choose some action that reaches a state discovered in the ($i$-1)-th iteration with positive probability. Such a state exists by construction, and the choice ensures that the initial strategy reaches the set of target or sink states almost surely.

When a Maximizer strategy is fixed, the main loop of Algorithm \ref{alg:si} solves the induced MDP $\game[\straa]$ (Line~\ref{line:si_mdp}). 
It need not remember the Minimizer strategy, but only uses the computed value estimates $\lb$ to greedily update Maximizer's strategy (Line~\ref{line:si_update}); note that here $\lb(\state,\action)$ is again computed from $\lb(\state)$ as in Equation~\ref{eq:Vsa}.
The algorithm stops when the Maximizer strategy does not change any more in one iteration. We can then compute the values and the corresponding Minimizer strategy by solving the induced MDP~ $\game[\straa]$.

\begin{algorithm}[htbp]
	
	\caption{Strategy iteration}\label{alg:si}
	\begin{algorithmic}[1]
		\Procedure{SI}{}
		\State $\straa' \gets$ arbitrary Maximizer \emph{attractor strategy}\hspace{1.5em} \Comment{Proper initial strategy} \label{line:si_attractor}
		\Repeat
		\State $\straa \gets \straa'$
		\For {$\state \in \states$}
		\State $\lb(\state) \gets \inf_{\strab} \pr_{s}^{\straa,\strab}(\Diamond \fstates)$ \hspace{4.85em} \Comment{Solve MDP for opponent} \label{line:si_mdp}
		\EndFor
		\For {$\state \in \states<\Box>$}
		\State $\straa'(\state) \gets \argmax_{\action \in \Av(\state)} \lb(\state,\action) $ \hspace{1.5em}\Comment{Greedily optimise choices} \label{line:si_update}
		\EndFor
		\Until{$\straa = \straa'$}
		\EndProcedure
	\end{algorithmic}
\end{algorithm}

\subsection{Quadratic Programming}

\emph{Quadratic programming (QP)} ~\cite{DBLP:conf/dimacs/Condon90} works by encoding the graph of the SG in a system of constraints. The only global (and local) optimum of the objective function under these constraints is 0, and it is attained if and only if the variable for every state is set to the value of that state~\cite[Section 3.1]{DBLP:conf/dimacs/Condon90}.
The proof as well as the construction of the quadratic program relies on the game being in a certain normal form, which consists of four conditions:
\begin{itemize}
	\item $\twoSucc$: For all $\state \in \states : |\Av(\state)| \leq 2$.
	\item $\noOneSucc$: If $|\Av(\state)| = 1$, then $\state$ is a target or a sink.
	\item $\halfProbs$: For all $\state, \state' \in \states , \action \in \Av(\state): \trans(\state,\action,\state') \in \{0, 0.5, 1\}$.
	\item $\stopping$: There are no ECs in $\game$ (except for the sinks and targets).
\end{itemize}
We shortly discuss the advantage of each condition of the normal form:
the reason for $\twoSucc$ and $\noOneSucc$ is that the objective function of the QP requires every (non-sink and non-target) state to have exactly 2 successors. Arguing about a game with average nodes instead of actions mapping to arbitrary probability distribution simplified the proofs, which is the advantage of $\halfProbs$. 
$\stopping$ was necessary, because of the problem of non-convergence in end components, as in Example \ref{ex:ubNoConvergeVI}.
We recall the procedure from \cite{condonComplexity} to transfer an arbitrary SG into a polynomially larger SG in normal form in \ifarxivelse{Appendix~\ref{app:cnf}}{the full version~\cite[Appendix A]{techreport}}.

We now state the quadratic program for an SG in normal form as given in \cite{DBLP:conf/dimacs/Condon90}, but adjusted to our notation.
Intuitively, the objective function is 0 if all summands are 0. And the summand for some state $\state$ is 0 if its value $\val(\state)$ is equal to the value of one of its actions $\val(\state,\action)$ or $\val(\state,\mathsf b)$. The program is quadratic, since all states (except targets and sinks) have exactly two successors and hence the summand for every state is a quadratic term.
The constraints encode the game, ensuring that Maximizer/Minimizer states use the action with the highest/lowest value and fixing the values of targets and sinks.
Note the additional definition of $\val(\state,\action)$, which assumes that all occurring non-trivial probabilities are $\sfrac{1}{2}$.

\begin{equation*}
\begin{array}{ll@{}ll}
\text{minimize}  & \displaystyle\sum\limits_
{\substack{\state \in \states \\ \Av(\state) = \{\action, \actionb\}}}
(\val(\state) - \val(\state,\action)) (\val(\state) - \val(\state,\actionb))&&\\
\text{subject to}& 
\displaystyle\val(\state) \geq \ \val(\state,\action) &\forall \state \in \states<\Box>: |\Av(\state)| = 2, \forall \action \in \Av(\state) &\\
& \displaystyle\val(\state) \leq \ \val(\state,\action)  &\forall \state \in \states<\circ>: |\Av(\state)| = 2, \forall \action \in \Av(\state) &\\
& \displaystyle\val(\state) = 1  &\forall \state \in \fstates &\\
& \displaystyle\val(\state) = 0  &\forall \state \in \sink &\\
\end{array}
\end{equation*}
\begin{equation*}
\text{where}\; \; \val(\state,\action) = \left\{\begin{array}{ll}
\val(\mathsf{s} ') & \text{for } |\post(\state, \action)| = \{\mathsf{s} ' \}\\
\sfrac{1}{2} \val(\mathsf{s} ') + \sfrac{1}{2} \val(\mathsf{s} '') & \text{for } |\post(\state, \action)| = \{\mathsf{s} ', \mathsf{s} '' \}
\end{array}\right\}
\end{equation*}

\section{Improvements}

In this section, we first generalize QP to be applicable to arbitrary SGs, thereby omitting the costly transformations into the normal form.
Then, we identify a hyperparameter of SI.
Finally, we provide two optimizations that are applicable to all three algorithms.

\subsection{Quadratic Programming for General Stochastic Games}\label{sec:generalQP}

Every transformation into the normal form (see \ifarxivelse{Appendix \ref{app:cnf}}{\cite[Appendix A]{techreport}}) adds additional states or actions to the $\SG$.
Firstly, we want to change the constraints of the QP so that it can deal with arbitrary SGs; secondly, we want to avoid blowing up the SG, as the time for solving the QP depends on the size of the SG.

\subsubsection{$\twoSucc$}
In order to drop the constraint that every state has at most two actions, we can generalize the summand of a state $\state$ in the objective function to $\prod_{\action \in \Av(\state)} (\val(\state) - \val(\state,\action))$.
The resulting program is no longer quadratic, as for a state with $n$ actions now the objective function has order $n$.
Thus, in the experiments, we report the verification times of both (i) a higher-order optimization problem for the original SG as well as (ii) a QP for the SG that was transformed to comply with $\twoSucc$.

One more step is needed to ensure that the objective function still is correct: Recall that we want the only global optimum of the objective function to be 0, and it should be attained if and only if $\val(\state) = \val(\state,\action)$ for some action $\action \in \Av(\state)$. However, if a Minimizer state has an odd number of actions, its summand in the objective function could be negative.
For example, a Minimizer state with three actions could choose its value to be smaller than all three actions. Multiplying three negative numbers results in a negative number, which is preferred to a summand of 0, since we minimize the objective function.
Thus, for a state with an odd number of successors, we duplicate the term for one of the actions.
Then the summand for every state is non-negative.
For a Maximizer state, all factors are greater or equal to 0, since by the first constraint of the QP we have $\val(\state) \geq \val(\state,\action)$ for every action. For a Minimizer state, either one of the factors is 0 or all factors are negative (by the second constraint of the QP), and multiplying an even number of negative numbers results in a positive number. 
Thus, 0 is the only global optimum of the objective function in the constrained region.

\subsubsection{$\noOneSucc$}
The $\noOneSucc$-constraint compels every state $\state \in \states$ that is not a target or a sink to have more than one action.
The transformation for complying with this constraint (see \ifarxivelse{Appendix \ref{app:no1act}}{\cite[Appendix A.2]{techreport}}) adds a second action to every state with only one action; however, the newly added action is chosen in such a way that it does not influence the value of the state, and can thus also be omitted (see \ifarxivelse{Appendix \ref{app:proofNo1Act}}{\cite[Appendix B.1]{techreport}} for the formal argument).
For a non-absorbing state $\state$ with only one action $\action$, we can simplify the program by not including it in the objective function, but only adding a single constraint $\val(\state) = \val(\state,\action)$.

\subsubsection{$\halfProbs$}
To comply with the $\halfProbs$-requirement, every transition $\trans(\state, \action, \state')$ with $\state, \state' \in \states$ and $\action \in \Av(\state)$ must have a probability of 0, 0.5 or 1. 
We can omit this constraint if we use the general definition of $\val(\state,\action)$ as in Equation \ref{eq:Vsa}, summing the successors with their respective given transition probability (see \ifarxivelse{Appendix \ref{app:halfProbs}}{\cite[Appendix B.2]{techreport}} for the formal argument).

\subsubsection{Stopping Game}\label{sec:QP_stopping}
The final and most complicated constraint of the normal form is that we require the $\SG$ to be stopping.
The transformation of an arbitrary game into a stopping one (see \ifarxivelse{Appendix \ref{app:stopping}}{\cite[Appendix A.4]{techreport}}) adds a transition to a sink with a small probability $\varepsilon$ to every action, thus ensuring that a sink (or a target) is reached almost surely. 
This is problematic not only because the added transitions blow up the quadratic program, but even more so because of the following:


The $\varepsilon$-transitions modify the value of the game. Theoretically, this is no problem, because the value is rational and we know the greatest possible denominator $q$ it can have. Thus, by choosing $\varepsilon$ sufficiently small, we ensure that the modified value does not differ by more than $\sfrac{1}{q}$ and we can obtain the original value by rounding.
However, practically, this denominator becomes smaller than machine precision even for small systems, resulting in immense numerical errors.
The $\varepsilon$ has to be strictly smaller%
\footnote{It actually has to be a lot smaller, since the \emph{value} of every state may differ by at most that amount, but this conservative upper bound suffices to prove our point.} 
than
 $(\sfrac{1}{4})^{\abs{\states}}$~\cite{condonComplexity}.
According to IEE 754.2019 standard\footnote{\url{https://standards.ieee.org/content/ieee-standards/en/standard/754-2019.html}}, the commonly used double machine precision is $10^{-16}$, so already for 27 states the necessary $\varepsilon$ becomes smaller than machine precision.
Thus, the transformation to a stopping game is inherently impractical.

%

\medskip

Our approach introduces additional constraints for every maximal end component (MEC) to ensure that the QP finds the correct solution.

For MECs where all states belong to the same player, the solution is straightforward.
All states in MECs with only Minimizer states have a value of 0, as they can choose to remain forever in the EC and not reach the target; they can be identified a priori and are part of the set of sinks $\sink$.
All states in MECs with only Maximizer states have the value of the best exit from that MEC~\cite{KKKW18}. Thus, for all $T \in \mec(\game)$ with $T \cap \states<\circ> = \emptyset$ we can introduce an additional constraint: $\forall \state \in T: \val(\state) = \bE<\val>(T)$, where $\bE$ is defined as for BVI (see Section \ref{sec:bvi}).
Note that max-constraints are expressed through continuous and boolean variables and therefore, the resulting quadratic program is a mixed-integer program.

For MECs containing states of both players, the values of the states depend on the best exit that Maximizer can ensure reaching against the optimal strategy of Minimizer. 
We cannot just set the value to the best exit of the whole MEC, as Minimizer might prevent some states in the MEC from reaching that best exit.
The solution of~\cite{KKKW18} to analyse the graph and figure out Minimizer's decisions on the fly is not possible, because we have to give the constraints a priori.

We solve the problem as follows: iterate over all strategy-pairs $(\straa,\strab)$ in the MEC and for each pair describe the corresponding value of every state $\valstra(\state)$ depending on the values of the exiting actions  $\exits = \{(\state,\action) \mid \state \in T \wedge \post(\state,\action) \nsubseteq T \}$.
Then constrain the value for all states in the MEC according to the optimal strategies, i.e. $\val(\state) = \max\limits_{\straa}\min\limits_{\strab}\valstra(\state)$.
This ensures that the value of every state is set to the best exit it can reach, because the optimal strategies are chosen.

It remains to define how to describe $\valstra$ depending on the exits $\exits$. 
For a pair of strategies $(\straa,\strab)$ in the MEC, we consider the induced Markov chain $\game[\straa,\strab]$. 
We modify the MC and let every state-action pair $(\state,\action) \in \exits$ lead to a sink state $e_{(\state,\action)}$. 
Then, for every such sink state, we compute the probability $p^{\straa,\strab}_{e_{(\state,\action)}}(t)$ to reach it from every state $t \in T$ by solving the MC.
Then we set $\displaystyle \valstra(t) = \sum_{(\state,\action) \in \exits} p^{\straa,\strab}_{e_{(\state,\action)}}(t) \cdot \val(\state,\action)$.
We summarize the procedure we just described in Algorithm \ref{alg:qp}.
We use \emph{all strategies on $T$} to denote every possible mapping that maps every state $t \in T$ to some available action $a \in \Av(t)$.

\begin{algorithm}[htbp]
	
	\caption{Algorithm to add constraints for MECs containing states of both players}\label{alg:qp}
	\begin{algorithmic}[1]
		\Procedure{ADD\_MEC\_CONSTRAINTS}{MEC $T \subseteq S$}
		\For {all strategies $\straa$ on $T_\Box$}
			\For {all strategies $\strab$ on $T_\circ$}
				\For {every $(s,a) \in \exits$}
					\State Compute $p^{\straa,\strab}_{e_{(\state,\action)}}(t)$ for all $t \in T$ by solving the modified induced MC $\game[\straa,\strab]$
				\EndFor
			\EndFor
		\EndFor
		
		\For {every $t \in T$}
			\State Add constraint: $\val(t) = \max\limits_{\straa}\min\limits_{\strab} \sum_{(\state,\action) \in \exits} p^{\straa,\strab}_{e_{(\state,\action)}}(t) \cdot \val(\state,\action)$
		\EndFor
		\EndProcedure
	\end{algorithmic}
\end{algorithm}

Note that $\val(\state,\action)$ is defined as usual (see Equation \ref{eq:Vsa}). As by definition of being an exit it depends on some successor $\state' \notin T$, the states in the MEC cannot depend only on each other any more, but they have to depend on exiting actions. This ensures there is a unique solution.
See \ifarxivelse{Appendix \ref{app:proofMEC}}{\cite[Appendix B.3]{techreport}} for the formal proof.

For a MEC of size $n$ we have to examine at most $n^{(\max_{\state \in T} \abs{\Av(\state)})}$ pairs of strategies, because it suffices to consider memoryless deterministic strategies. 
Since the $\min$ and $\max$ constraints have to be explicitly encoded, we have to add a number of constraints that is exponential in the size of the MEC. 

\medskip

In summary, we have shown how to replace every condition of the normal form by modifying the constraints and objective function of the program.
The modifications always ensure that the program still computes the correct value, because it still only has a single global optimum in the constrained region, namely if every state variable is set to its value. 
Thus, we can provide a higher-order program to solve arbitrary SGs, and a quadratic program to solve SGs that only have to satisfy the $\twoSucc$ condition.

\subsection{Opponent Strategy for Strategy Iteration}

We can tune the MDP solution method that is used to compute the opponent strategy in Line \ref{line:si_mdp} of Algorithm \ref{alg:si}.
We need to ensure that we fix the new choices of Maximizer correctly.
For this, we can use the precise MDP solution methods strategy iteration or linear programming (LP, a QP with an objective function of order 1). 
However, we do not need the precise solution of the induced MDP, but it suffices to know that an action is better than all others.
So we can also use bounded value iteration and check that the lower bound of one action is larger than the upper bound of all other actions, and thus we can stop the algorithm earlier. 
This approximation and the fact that VI tends to be the fastest methods in MDPs can speed up the solving.
Using unguaranteed VI is dangerous, as it might stop too early and return a wrong strategy.

\subsection{Warm Start}
All three solution methods can benefit from prior knowledge.
VI and quadratic/higher-order programs (QP/HOP) can immediately use initial solution vectors obtained by domain knowledge or any precomputation.
For VI, it is necessary to know whether it is an upper or lower estimate to use the prior knowledge correctly.
The QP/HOP optimization process can start at the given initial vector.

SI can use the information of an initial estimate to infer a good initial strategy, as already suggested in \cite{KM17}. This reduces the number of iterations of the main loop and thus the runtime.
However, we have to ensure that the resulting strategy is proper. For example, we can check whether the target and sink states are reached almost surely from every state, and if not, we change the strategy to an attractor strategy where necessary, preferring those allowed actions that have a higher value.

We can also improve the MDP-solving for SI (Line \ref{line:si_mdp} of Algorithm \ref{alg:si}) by giving it the knowledge we currently have. 
We anyway save the estimate $\lb$ of the previous iteration and, since the strategies of Maximizer get monotonically better, $\lb$ certainly is a lower bound for the values in the MDP.

Even in the absence of domain knowledge or sophisticated precomputations, we can run unguaranteed VI first in order to get some estimates of the values. 
Then we can use those estimates for SI and QP/HOP.
Note that this is similar in spirit to the idea of optimistic value iteration~\cite{OVI}: utilize VI's ability to usually deliver tight lower bounds and then verify them.

\subsection{Topological Improvement}

For MDPs, topological improvements have been proposed for VI~\cite{TVI1} and for SI~\cite[Algorithm 3]{KM17}.
These utilize the fact that the underlying graph of the MDP can be decomposed into a directed acyclic graph of strongly connected components (SCC).
Intuitively speaking, there are parts of the graph that, after leaving them, can never be reached again. 
Their value depends solely on the parts of the state space that come after them.
So instead of computing the values on the whole game at once, one can iterate over the SCC in a backwards fashion, starting with the target and sink states and then propagating the information and solving the SCCs one by one according to their topological ordering.

This idea was extended to BVI for MDPs in \cite{ensure}. The proof generalizes to SGs, as the basic argument of the topological ordering of SCCs is independent from introducing a second player.
As the proof relies on the solutions for the later components being $\varepsilon$-precise, solving those components with the precise methods SI or QP is of course also possible.

\section{Experimental Results}

\textbf{Implementation:} We implemented all our algorithms as an extension of PRISM-games~3.0~\cite{PRISM-games3}. 
They are available via github \url{https://github.com/ga67vib/Algorithms-For-Stochastic-Games}.
For BVI, we reimplemented the algorithm as described in~\cite{KKKW18}; for SI and QP, this is the first implementation in PRISM-games.
To solve the quadratic program, we used Gurobi\footnote{\url{https://www.gurobi.com/}} or CPLEX\footnote{\url{https://www.ibm.com/analytics/cplex-optimizer}}.
For the higher-order programs, we constructed them with AMPL and solved them with MINOS\footnote{\url{https://ampl.com/products/solvers/solvers-we-sell/minos/}}.

\textbf{Setup:} All experiments were conducted on a Linux~Manjaro server with a 3.60~GHz Intel(R) Xeon(R) W-2123 CPU and 64~GB of RAM. We used a timeout of 15 minutes and set the java heap size for PRISM-games to 32~GB and the stack size to 16~GB\footnote{-javamaxmem 32g -javastack 16g}.
The precision for BVI was set to $10^{-6}$.
In theory, the other algorithms are precise; in practice, QP and higher-order programming can have numerical problems.
For a discussion about the practical precision of SI we refer to the end of Section \ref{sec:comparison} as well as \ifarxivelse{Appendix~\ref{app:practicalSI}}{\cite[Appendix C.3]{techreport}}.

\textbf{Case studies:} As case studies, we used those that are distributed with PRISM-games 3.0 as well as some handcrafted corner cases; see \ifarxivelse{Appendix \ref{app:models}}{\cite[Appendix C.1]{techreport}} for a detailed description.

We first analyze the impact of our optimizations by comparing different variants of the same algorithm.
The full tables that this analysis is based on are in \ifarxivelse{Appendix \ref{app:tables}}{\cite[Appendix C.2]{techreport}}, but general trends are also visible in Table \ref{tab:main}. 
Based on this, we select the best variants of each algorithm and compare between the algorithms.

\subsection{Value Iteration}
We could reproduce most of the findings of \cite{KKKW18}: the overhead of BVI compared to unguaranteed VI is usually negligible and not performing the expensive deflate operation in every step speeds up the computation.
Unguaranteed VI fails on three of the models.
In contrast to \cite{KKKW18}, we found no model where the simulation based asynchronous version BRTDP was significantly faster than BVI. In fact, BRTDP is only faster on a single model (cloud6, where BVI takes 2 seconds), but significantly slower on many others, often even failing to produce results in time.
Note that the implementation of BRTDP was in PRISM-games 2, and thus the disadvantage may also have technical reasons; improvements in the simulation engine or data structures might lead to speed-ups that make BRTDP competitive again.

The new topological variant of BVI (called TBVI) is usually in the same order of magnitude as the default approach.
The exception to this are the models AV15\_15 and especially MulMec, where TBVI is a lot slower.
A possible explanation is that TBVI solves every SCC of the model with a precision of $\varepsilon$. So when one SCC is solved and has a difference of almost exactly $\varepsilon$ between upper and lower bound, SCCs depending on it take a longer time to converge, as the information about their exits is suboptimal. This is particularly problematic when the model has a structure like MulMec, namely a chain of MECs, and hence a chain of SCCs.
This problem is not specific to topological VI for SGs, but can also occur for MDPs.

\subsection{Strategy Iteration}
Using BVI for the opponent's MDP and the warm start usually lead to small speed-ups. We did not consider using linear programming for the opponent's MDP, as it is not supported by PRISM.
The topological variant is significantly better in two\_inv and MulMec\_e3, but on the rest of the models performs very similar to the non-topological version. 
Combining topological SI and BVI for the opponent's MDP leads to the same problems with MulMec as when using topological BVI.
In contrast, topological SI with SI for the opponent's MDP works, because SI solves the SCCs precisely, allowing the SCCs depending on the previous solutions to converge as well.

\subsection{Quadratic Programming}
The original version of quadratic programming~\cite{DBLP:conf/dimacs/Condon90}, that requires to transform the SG into normal form is impractical, producing a result within 15 minutes for only 3 of 34 case studies, and even there taking a lot more time than the improved version.
After dropping the constraints $\noOneSucc$ and $\halfProbs$, it can correctly solve 14 of the case studies in time.
Both of these variants are prone to the numerical errors described in Section \ref{sec:QP_stopping} and produce incorrect results on some models.
The quadratic program obtained after dropping all but the $\twoSucc$ constraint is solved successfully by Gurobi in 19 instances; CPLEX on the other hand only solves 7 instances correctly, one time even reporting an incorrect result.
The warm start helps Gurobi on the model charlton1. Several other times, Gurobi discards the given initial suggestion and uses its own heuristics; thus, the warm start sometimes incurs a slight overhead.

Dropping all constraints, the higher-order program (HOP) with the solver Minos gets the correct result on 23 instances.
The HOP is typically faster than the QP, except on the models HW and AV.
The topological variant of both the quadratic programs as well as the higher-order program can lead to significant speed-ups, for example on charlton1, mdsm1, two\_inv or HW10\_10\_2.
The topological HOP is strictly better than all other algorithms in this subsection.

To estimate the impact of the EC solution method, we used several handcrafted or modified models:
A single large MEC (BigMec\_e2, with a MEC of size 201) cannot be solved with our approach, as there are too many choices; possibly, some heuristic could help identify reasonable strategies.
In contrast, small MECs do not affect runtime a lot.
The models cdmsn and dice50 prepended with a single three-state MEC are solved in the same time as the original models.
Even a chain of 1000 three-state MECs can be solved quickly (MulMec\_e3).

\subsection{Comparison}\label{sec:comparison}

\begin{table}[]
	\caption{Table for the verification times (in seconds) of several variations of the algorithms.
		An X in the table denotes that the computation did not finish within 15 minutes.
		A \colorbox{red!25}{red background colour} indicates that the returned result was wrong.		
		The four left-most columns give information about case study: its name, its size, the maximum/average number of actions and the number of relevant MECs.
		This table shows a selection of the most interesting case studies. They are roughly sorted by increasing size/difficulty, with scaled versions of the same model grouped together.
		The considered algorithms are the default and topological variant of BVI (with deflating only happening every 100 iterations); SI with BVI as MDP solver and topological SI with warm start and SI as MDP solver; and QP with Gurobi as solver and warm start, as well as the topological higher order program.
	}
	\label{tab:main}
\begin{tabular}{l rrr | rr rr rr}
Case Study  & States & Acts    & MECs & BVI$_{100}$     & TBVI$_{100}$    & SI$_\text{BV}$       & TSI$_\text{SI}^\text{W}$ & QP$_\text{G}^\text{W}$ & THOP         \\ \toprule

prison\_dil & 102    & 3/1.34  & 0    & \textless{}1 & \textless{}1 & \textless{}1 & \textless{}1                & 8                         & \textless{}1 \\
charlton1   & 502    & 3/1.56  & 0    & \textless{}1 & \textless{}1 & \textless{}1 & \textless{}1                & 144                       & \textless{}1 \\
cdmsn       & 1,240   & 2/1.66  & 0    & \textless{}1 & \textless{}1 & \textless{}1 & \textless{}1                & \textless{}1              & \textless{}1 \\
cdmsnMec    & 1,244   & 2/1.66  & 1    & \textless{}1 & \textless{}1 & \textless{}1 & \textless{}1                & \textless{}1              & \textless{}1 \\
cloud6      & 34,954  & 13/4.45 & 2176 & 2            & 3            & \textless{}1 & \textless{}1                & X                         & X            \\
mdsm1       & 62,245  & 2/1.34  & 0    & 5            & 3            & 5            & 3                           & X                         & 3            \\
dice50      & 96,295  & 2/1.48  & 0    & 6            & 6            & 6            & 6                           & X                         & 6            \\
dice50Mec   & 96,299  & 2/1.48  & 1    & 6            & 6            & 7            & 6                           & X                         & 6            \\
two\_inv    & 172,240 & 3/1.34  & 0    & 13           & 13           & 19           & 13                          & X                         & 20           \\
HW10\_10\_1 & 400,000 & 5/2.52  & 0    & 10           & 11           & 11           & 11                          & 98                        & 11           \\
HW10\_10\_2 & 400,000 & 5/2.52  & 0    & \textless{}1 & 1            & 2            & 2                           & 49                        & 1            \\
AV10\_10\_1 & 106,524 & 6/2.17  & 0    & \textless{}1 & \textless{}1 & \textless{}1 & \textless{}1                & 3                         & \textless{}1 \\
AV10\_10\_2 & 106,524 & 6/2.17  & 6    & 72           & 70           & 79           & 77                          & X                         & X            \\
AV10\_10\_3 & 106,524 & 6/2.17  & 1    & 45           & 55           & 50           & 60                          & X                         & X            \\
AV15\_15\_1 & 480,464 & 6/2.14  & 0    & 1            & 1            & 2            & 2                           & 11                        & 2            \\
AV15\_15\_2 & 480464 & 6/2.14  & 6    & X            & X            & 825          & X                           & X                         & X            \\
AV15\_15\_3 & 480,464 & 6/2.14  & 1    & X            & X            & 500          & X                           & X                         & X            \\
hm\_30      & 61     & 1/1.00  & 0    & X            & X            & X            & X                           & \textless{}1              & \colorbox{red!25}{\textless{}1} \\
MulMec\_e2  & 302    & 2/1.99  & 100  & 2            & X            & X            & \textless{}1                & \textless{}1              & \textless{}1 \\
MulMec\_e3  & 3,002   & 2/2.00  & 1000 & 151          & X            & X            & 3                           & 7                         & 4            \\
BigMec\_e2  & 203    & 2/1.99  & 1    & \textless{}1 & \textless{}1 & \textless{}1 & \textless{}1                & X                         & X            \\
BigMec\_e3  & 2,003   & 2/2.00  & 1    & 1            & 1            & 1            & 1                           & X                         & X            \\
BigMec\_e4  & 20,003  & 2/2.00  & 1    & 226          & 230          & X            & X                           & X                         & X           \\
\bottomrule
\end{tabular}
\end{table}

Comparing the algorithms, we see that BVI and SI perform very similar. BVI succeeds in the largest version of BigMec, but SI is the only one to solve the large and complicated models AV15\_15\_2/3, and TSI has the best runtime in MulMec\_e3.
TSI benefits from a model with small subcomponents that it can quickly solve (as in MulMec), while BVI is fast in subgraphs without probabilistic cycles (as the large chains in BigMec).
Depending on the model structure, both of these algorithms are a viable choice.

Topological higher-order programming (THOP), the improved version of QP, is comparable on many case studies, but still a lot worse on several others, e.g. cloud6 and BigMec. This volatility is even more pronounced for the QP.
The reason for this can be MECs which blow up the QP, but it can also happen in models with few or no MECs; in the latter case, we do not know which property of the model slows down the solving.
Note that QP and HOP are able to solve models with many small MECs (MulMec\_e3) quickly, while already one medium sized MEC (BigMec\_e2) makes it infeasible.

The model hm\_30 from \cite{BVI} deserves special attention: It was handcrafted as an adversarial example for VI.
Moreover, since the solvers for MDPs and MCs in PRISM use variants of VI, and since SI relies on those solvers, the PRISM implementation of SI also fails on the model\footnote{See \ifarxivelse{Appendix \ref{app:practicalSI}}{\cite[Appendix C.3]{techreport}} for more details on this.}.
QP shines on this model, being the only method to solve it.
However, for increasing parameter $N$, the probabilities in hm\_$N$ become so small that they are at the border of numerical stability.
If the state-chains in the model were prolonged by one more state (i.e. the parameter $N$ is set to 31), QP has numerical problems and reports an incorrect result.
Similarly, noting that THOP reports an incorrect result on hm\_30, one can experimentally find out that THOP succeeds only for $N\leq 25$.
So if the model exhibits very small probabilities, one has to consider using a solver capable of arbitrary-precision arithmetic.

\section{Conclusions} \label{sec:concl}

We have extended the three known classes of algorithm -- value iteration, strategy iteration and quadratic programming -- with several improvements and compared them both theoretically and practically. 

In summary, for all algorithms, the structure of the underlying graph is more important than its size; thus knowledge about the model is relevant both for estimating the expected time, as well as the preferred algorithm and combination of optimizations.
BVI and SI perform very similar on most models in our practical evaluation; each of them has some models where they are better.
Quadratic/higher order programming is volatile and typically slower than the other two; however, the used solver has a huge impact, as we already see when changing between CPLEX, Gurobi and Minos. 
Thus, advances in the area of optimization problems could make this solution method the most practical.

A direction for future work is to extend all algorithms to other objectives, e.g. total expected reward, mean payoff or parity.
Further, coming up with a polynomially sized convex QP would result in a polynomial-time algorithm, solving the long-standing open question.

\newpage
\bibliographystyle{eptcs}
\bibliography{ref}

\newpage
\ifarxivelse{
\appendix
\section{Transforming an Arbitrary Stochastic Game into Normal Form}\label{app:cnf} 

We describe the constructions of \cite{condonComplexity} to transform an arbitrary SG into one in normal form that is only polynomially larger. Normal form requires satisfying four conditions:
\begin{itemize}
	\item ($\twoSucc$) For all $\state \in \states : |\Av(\state)| \leq 2$.
	\item ($\noOneSucc$) If $|\Av(\state)| = 1$, then $\state$ is a target or a sink.
	\item ($\halfProbs$) For all $\state \in \states , \action \in \Av(\state), \state' \in \states: \trans(\state,\action,\state') \in \{0, 0.5, 1\}$.
	\item $\stopping$: There are no ECs in $\game$ (except for the sinks and targets).
\end{itemize}

\subsection{$\twoSucc$}
In normal form, every state $\state$ that is not an absorbing state must have at most two actions. To transform an arbitrary $\SG$ into one complying with $\twoSucc$, take every state $\state$ with $|\Av(\state) > 2|$ and construct a binary tree as illustrated in Figure \ref{ex:2Succ}. Two actions of $\state$ are taken and given to a new vertex $v'$. State $\state$ has then one action leading to the additional state $v'$ instead of its previous two actions. This can be done iteratively until there is a binary tree where $\state$ is the root and has only two actions. Note that every other inner node in this tree also has two actions, and the leaves are exactly $v_1, v_2, ..., v_k$. If a state $\state$ has $n \geq 2$ actions, then after the transformation there are $n-2$ additional states.
Every newly added state belongs to the same player as the original state $\state$.
\begin{figure}[htbp]

\centering
\begin{tikzpicture}[thick,scale=0.9, every node/.style={font=\large, transform shape}]
\drawdummy (init) at (2,1) {};
\drawbox (q) at (2,0) {$\mathsf{s_0}$};
\drawdummy (mid) at (2,-2) {};
\drawbox (11) at (4,-3) {$\mathsf{\target}$};
\drawbox (12) at (2,-3) {$\mathsf{s_1}$};
\drawcirc (0) at (0,-3)  {$\mathrm{\zeroSink}$};

\drawdummy (l) at (4.5, -1) {};
\drawdummy (r) at (5.5, -1) {};

\draw[->] (init) to (q);
\draw[-]  (q) to node [left ,midway] {$\mathsf{b}$}(mid);
\draw[->] (mid) to node [midway, above] {$\frac{1}{2}$} (0);
\draw[->] (mid) to node [midway, right] {$\frac{1}{2}$} (12);
\draw[->]  (q) to node [left ,midway] {$\mathsf{a}$}(0);
\draw[->]  (q) to node [above ,midway] {$\mathsf{c}$}(11);

\draw[->]  (0) to[loop below]  node [midway,below] {$\mathsf{a}$} (0);
\draw[->]  (11) to [loop below] node [midway,below] {$\mathsf{a}$} (11);
\draw[->]  (12) to [loop below] node [midway,below] {$\mathsf{a}$} (12);

\draw[->] (l) to (r);
\end{tikzpicture}
\begin{tikzpicture}[thick,scale=0.9, every node/.style={font=\large, transform shape}]
\drawdummy (init) at (2,1) {};
\drawbox (q) at (2,0) {$\mathsf{s_0}$};
\drawbox (p) at (1,-1.5){$\mathsf{s_1'}$}; 
\drawdummy (mid) at (1,-3) {};
\drawbox (11) at (4,-4) {$\target$};
\drawbox (12) at (2,-4) {$\mathsf{s_1}$};
\drawcirc (0) at (0,-4)  {$\mathrm{\zeroSink}$};

\draw[->] (init) to (q);
\draw[-]  (p) to node [right ,pos=0.3] {$\mathsf{b}$}(mid);
\draw[->] (mid) to node [midway, above] {$\frac{1}{2}$} (0);
\draw[->] (mid) to node [midway, above] {$\frac{1}{2}$} (12);
\draw[->]  (p) to node [left ,midway] {$\mathsf{a}$}(0);
\draw[->]  (q) to node [right ,midway] {$\mathsf{b}$}(11);
\draw[->]  (q) to node [left ,midway] {$\mathsf{a}$}(p);

\draw[->]  (0) to[loop below]  node [midway,below] {$\mathsf{a}$} (0);
\draw[->]  (11) to [loop below] node [midway,below] {$\mathsf{a}$} (11);
\draw[->]  (12) to [loop below] node [midway,below] {$\mathsf{a}$} (12);

\end{tikzpicture}
\caption[Transformation into $\twoSucc$]{An example of transforming an $\SG$ into one fulfilling the $\twoSucc$-constraint. $s_0$ has more than two actions, so a binary subtree is built up where each states except for the leafs $\target, s_1, \zeroSink$ has always exactly two actions.}
\label{ex:2Succ}
\end{figure}
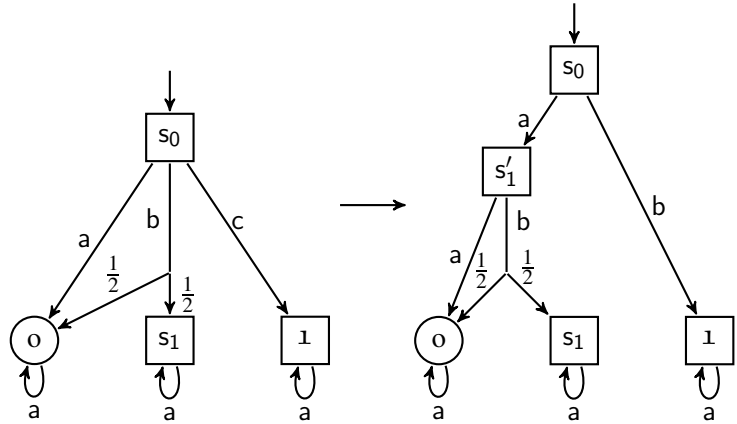

\subsection{$\noOneSucc$}\label{app:no1act}
Every state that is not a target or sink must have more than one action. We can add a second action to every state that is missing one. This action leads to a target in the case of $\state \in \states<\circ>$ and otherwise to a sink. If there is no sink in $\states$, we can introduce an artificial sink. In any optimal strategy, neither player would choose the additional action, as the other option is at least as good as the additional one. Therefore, the additional actions do not influence the value of any state.

In Figure \ref{ex:exampleCondonBad}, $s_0$ has an additional action leading to $\zeroSink$. Every state $\state \in (\states<\circ> \setminus \{\zeroSink\})$ should have an action leading to $\target$, but we omit this to improve the readability of the figure.

\subsection{$\halfProbs$}
The normal form requires that transition probabilities have either are either 0, 0.5 or 1. This implies that every action has either one or two successors.
Let $\state \in \states$ be a state which has an action $\action$ that has an arbitrary amount of successors $v_1, v_2, ..., v_k$ with rational transitions probabilities $p_i:= \trans(\state, \action\ v_i) \in [0,1] \subset \mathbb{Q}$. 
Consider the greatest common divisor $q$ of all occurring transition probabilities of ($\state, \action$). Let $q'$ be the smallest power of 2 such that $q' \geq q$, i.e. $q' = 2^k, k \in \mathbb{N}: 2^{k-1} < q < 2^k = q'$. 

Create $\sfrac{1}{2} \cdot q'$ new vertices, each with one action and two transitions with probability 0.5. Out of the $q'$ many transitions, $p_i \cdot q$ lead to $v_i$. If a vertex has two transitions assigned that lead to the same state, the two transitions are unified to one with transition probability 1.  The $q' - q$ remaining transitions lead to $\state$. From the new vertices, we build up a binary tree with $\state$ as root such that $\state$ and every new state have one action each with two transitions that have only transition probabilities of 0, 0.5 or 1, and such that the probability of reaching $v_i$ from $\state$ is $p_i$.

Figure \ref{ex:exampleCondonBad} illustrates an example. $s_1$ has an action $a$ with transition probabilities $\sfrac{3}{9}, \sfrac{1}{9}, \sfrac{5}{9}$. The common divisor $q$ of all these fractions is 9. The next power of 2 is 4, so we have $q = 9 \leq 16 = 2^4 = q'$. We need 8 states $h_1'', h_2'', ..., h_8''$ that store the 16 transitions. Every transition has a probability of $\sfrac{1}{2}$.
\begin{itemize}
	\item The probability to reach $\zeroSink$ from $\state$ is $\sfrac{3}{9}$, therefore 3 transitions have to lead there. We let $h_1''$ lead to $\zeroSink$ with probability 1 and $h_2''$ with one of the two transitions.
	\item The probability to reach $s_2$ from $\state$ is $\sfrac{1}{9}$, therefore the second transition of $h_2''$ leads in $s_2$.
	\item The probability to reach $\target$ from $\state$ is $\sfrac{5}{9}$, therefore $h_3''$, $h_4''$ and one transition of $h_5''$ lead to $\target$.
	\item The remaining transitions lead back to $\state$.
\end{itemize}
To connect $h_1'', h_2'', ..., h_8''$ to $\state$ a binary tree is constructed with $h_1', h_2', h_3', h_4'$ and $h_1, h_2$.

\begin{figure}[htbp]

\centering
\begin{tabular}{@{}c@{}}
\begin{tikzpicture}[thick,scale=0.9, every node/.style={font=\large, transform shape}]
\drawdummy (init) at (0,1) {};
\drawbox (q) at (0,0) {$\mathsf{s_0}$};
\drawcirc (p) at (2,0) {$\mathsf{s_1}$};
\drawdummy (mid) at (2,-1) {};
\drawbox (11) at (4,-3) {$\mathsf{\target}$};
\drawbox (12) at (2,-3) {$\mathsf{s_2}$};
\drawcirc (0) at (0,-3)  {$\mathrm{\zeroSink}$};

\drawdummy (l) at (2, -4.5) {};
\drawdummy (r) at (2, -5.5) {};

\draw[->] (init) to (q);
\draw[->]  (q) to node [above ,midway] {$\mathsf{a}$}(p);

\draw[-] (p) to node [left, midway] {$\mathsf{a}$}(mid);

\draw[->] (mid) to node [midway, above, pos=0.7] {$\frac{3}{9}$} (0);
\draw[->] (mid) to node [midway, above, pos=0.7] {$\frac{5}{9}$} (11);
\draw[->] (mid) to node [midway, left] {$\frac{1}{9}$} (12);

\draw[->]  (0) to[loop below]  node [midway,below] {$\mathsf{a}$} (0);
\draw[->]  (11) to [loop below] node [midway,below] {$\mathsf{a}$} (11);
\draw[->] (12) to [bend right] node [midway,below] {$\mathsf{a}$} (11);

\draw[->] (l) to (r);
\end{tikzpicture}
\end{tabular}
\vspace{\floatsep}

\centering
\begin{tabular}{@{}c@{}}
\begin{tikzpicture}[thick,scale=0.6, every node/.style={font=\huge, transform shape}]
\drawdummy (init) at (-8,2) {};
\drawbox (q) at (-8,0) {$\mathsf{s_0}$};
\drawcirc (p) at (4,0) {$\mathsf{s_1}$};
\drawdummy (mid) at (4,-2) {};
\drawbox (11) at (2,-18) {$\mathsf{\target}$};
\drawbox (12) at (-1,-18) {$\mathsf{s_2}$};
\drawcirc (0) at (-4,-18)  {$\mathrm{\zeroSink}$};

\draw[->] (init) to (q);
\draw[->] (q) to node [above, midway] {$\mathsf{a}$}(p);
\draw[->] (q) to [bend right] node [left, midway] {$\mathsf{b}$}(0);

\draw[->] (0) to[loop below]  node [midway,below] {$\mathsf{a}$} (0);
\draw[->] (11) to [loop below] node [midway,below] {$\mathsf{a}$} (11);
\draw[->] (12) to [bend left] node [midway,below] {$\mathsf{b}$} (0);
\draw[->] (12) to [bend right] node [midway,below] {$\mathsf{a}$} (11);

\drawcirc (sHalf1) at (-4, -12) {$\mathsf{h_1''}$};
\drawcirc (sHalf2) at (-2, -12) {$\mathsf{h_2''}$};
\drawdummy (sHalf2Mid) at (-2, -15) {};
\drawcirc (sHalf3) at (0, -12) {$\mathsf{h_3''}$};
\drawcirc (sHalf4) at (2, -12) {$\mathsf{h_4''}$};
\drawcirc (sHalf5) at (4, -12) {$\mathsf{h_5''}$};
\drawdummy (sHalf5Mid) at (4, -15) {};
\drawdummy (sHalf5Help) at (14,0) {};
\drawcirc (sHalf6) at (6, -12) {$\mathsf{h_6''}$};
\drawcirc (sHalf7) at (8, -12) {$\mathsf{h_7''}$};
\drawcirc (sHalf8) at (10, -12) {$\mathsf{h_8''}$};

\draw[->] (sHalf1) to node [left, midway] {$\mathsf{a}$}(0);
\draw[-] (sHalf2) to node [right, midway] {$\mathsf{a}$}(sHalf2Mid);
\draw[->] (sHalf2Mid) to node [text width =0.5cm, midway, left] {$\frac{1}{2}$} (0);
\draw[->] (sHalf2Mid) to node [text width =-0.5cm, midway, right] {$\frac{1}{2}$} (12);

\draw[->] (sHalf3) to node [text width =1.0cm, above, midway] {$\mathsf{a}$}(11);
\draw[->] (sHalf4) to node [text width =1.0cm, above, midway] {$\mathsf{a}$}(11);

\draw[-] (sHalf5) to node [text width =1.0cm, right, midway] {$\mathsf{a}$}(sHalf5Mid);
\draw[->] (sHalf5Mid) to node [text width =0.5cm, midway, left] {$\frac{1}{2}$} (11);
\draw[-] (sHalf5Mid) to (14, -15)  to node [text width =0.5cm, midway, right] {$\frac{1}{2}$} (sHalf5Help);
\draw[->] (sHalf5Help) to (p);

\draw[->] (sHalf6) to (8, -14) to (13, -14) to (13, -1) to node [text width =1.0cm, above, near start] {$\mathsf{a}$}(p);
\draw[->] (sHalf7) to (10, -13) to (12, -13) to (12, -2) to node [text width = -0.1cm, above, near start] {$\mathsf{a}$}(p);
\draw[->] (sHalf8) to (11, -12) to (11, -3) to node [text width =-0.1cm, above, near start] {$\mathsf{a}$}(p);

\drawcirc (sLevelTwo1) at (-3, -8) {$\mathsf{h_1'}$};
\drawdummy (sLevelTwo1Mid) at (-3, -10) {};
\drawcirc (sLevelTwo2) at (1, -8) {$\mathsf{h_2'}$};
\drawdummy (sLevelTwo2Mid) at (1, -10) {};
\drawcirc (sLevelTwo3) at (5, -8) {$\mathsf{h_3'}$};
\drawdummy (sLevelTwo3Mid) at (5, -10) {};
\drawcirc (sLevelTwo4) at (9, -8) {$\mathsf{h_4'}$};
\drawdummy (sLevelTwo4Mid) at (9, -10) {};

\draw[-] (sLevelTwo1) to node [text width =1.0cm, right, midway] {$\mathsf{a}$}(sLevelTwo1Mid);
\draw[->] (sLevelTwo1Mid) to node [text width =0.5cm, near start, left] {$\frac{1}{2}$} (sHalf1);
\draw[->] (sLevelTwo1Mid) to node [text width =-0.5cm, near start, right] {$\frac{1}{2}$} (sHalf2);

\draw[-] (sLevelTwo2) to node [text width =1.0cm, right, midway] {$\mathsf{a}$}(sLevelTwo2Mid);
\draw[->] (sLevelTwo2Mid) to node [text width =0.5cm, near start, left] {$\frac{1}{2}$} (sHalf3);
\draw[->] (sLevelTwo2Mid) to node [text width =-0.5cm, near start, right] {$\frac{1}{2}$} (sHalf4);

\draw[-] (sLevelTwo3) to node [text width =1.0cm, right, midway] {$\mathsf{a}$}(sLevelTwo3Mid);
\draw[->] (sLevelTwo3Mid) to node [text width =0.5cm, near start, left] {$\frac{1}{2}$} (sHalf5);
\draw[->] (sLevelTwo3Mid) to node [text width =-0.5cm, near start, right] {$\frac{1}{2}$} (sHalf6);

\draw[-] (sLevelTwo4) to node [text width =1.0cm, right, midway] {$\mathsf{a}$}(sLevelTwo4Mid);
\draw[->] (sLevelTwo4Mid) to node [text width =0.5cm, near start, left] {$\frac{1}{2}$} (sHalf7);
\draw[->] (sLevelTwo4Mid) to node [text width =-0.5cm, near start, right] {$\frac{1}{2}$} (sHalf8);

\drawcirc (sLevelOne1) at (-1, -4) {$\mathsf{h_1}$};
\drawdummy (sLevelOne1Mid) at (-1, -6) {};
\drawcirc (sLevelOne2) at (7, -4) {$\mathsf{h_2}$};
\drawdummy (sLevelOne2Mid) at (7, -6) {};

\draw[-] (sLevelOne1) to node [text width =1.0cm, right, midway] {$\mathsf{a}$}(sLevelOne1Mid);
\draw[->] (sLevelOne1Mid) to node [text width =1cm, near start, left] {$\frac{1}{2}$} (sLevelTwo1);
\draw[->] (sLevelOne1Mid) to node [text width =-1cm, near start, right] {$\frac{1}{2}$} (sLevelTwo2);

\draw[-] (sLevelOne2) to node [text width =1.0cm, right, midway] {$\mathsf{a}$}(sLevelOne2Mid);
\draw[->] (sLevelOne2Mid) to node [text width =1cm, near start, left] {$\frac{1}{2}$} (sLevelTwo3);
\draw[->] (sLevelOne2Mid) to node [text width =-1cm, near start, right] {$\frac{1}{2}$} (sLevelTwo4);

\draw[-] (p) to node [text width =1.0cm, above, midway] {$\mathsf{a}$}(mid);
\draw[->] (mid) to node [midway, above] {$\frac{1}{2}$} (sLevelOne1);
\draw[->] (mid) to node [midway, above] {$\frac{1}{2}$} (sLevelOne2);

\end{tikzpicture}
\end{tabular}
\caption[Transformation into \cnf ]{An example of an arbitrary $\SG$ that gets transformed into \cnf{}. States $h_1'', h_2'', ..., h_8''$, $h_1', h_2', h_3', h_4'$ and $h_1, h_2$ are introduced to fulfill the $\halfProbs$-constraint. State $\mathsf{s_2}$ has now a second action leading to sink $\zeroSink$, as otherwise it would not compy with $\noOneSucc$. All the $\circ$-states with only one action except the sink $\zeroSink$ must have an action leading to $\target$ but we omit these for the sake of readability.}
\label{ex:exampleCondonBad}
\end{figure}

\subsection{Stopping game}\label{app:stopping}
To avoid the possibility of never reaching any absorbing state, we add a transition with a small probability $\varepsilon$ leading to a sink-state $\zeroSink$ to every action. If the players pick strategies that trap the play in a MEC in the original $\SG$, the play would almost surely reach $\zeroSink$ in the modified $\SG$. 
If the $\varepsilon$ is chosen sufficiently small, one can infer the value in the original SG from the modified SG~\cite{condonComplexity}.
This is due to the fact that the value of an SG must be rational, where the denominator is at most $4^{\abs{\states}}$, and thus we can round the value in the modified SG to the nearest fraction with denominator $4^{\abs{\states}}$.

In Figure \ref{ex:exampleToStopping}, we illustrate the transformation of a $\SG$. Although the initial $\SG$ is already stopping in this example, we chose this simple game for explanatory purposes.
For the $\varepsilon$-transitions that have to be introduced to comply with the $\stopping$ constraint, the paper~\cite{DBLP:conf/dimacs/Condon90} suggest something smarter instead of constructing a binary tree:
for every state-action pair ($\state, \action$) add a new state $\state'$ which has only one action that has the same successors and transition probabilities as ($\state, \action$) in the initial game. The state $\state$ instead has a new action $\action'$ with $\sfrac{1}{2}$-probabilities of leading either to $\state'$ or to the state $h_{1}^{\state, \action}$ of a chain of $m \in \mathbb{N}$ (sufficiently large) new states $h_{1}^{\state, \action}, h_{2}^{\state, \action}, ... h_{m}^{\state, \action}$. Each state $h_{i}^{\state, \action}$ with $i \in [m-1]$ has only one action with two $\sfrac{1}{2}$-probabilities leading to either $h_{i+1}^{\state, \action}$ or $\state'$ (see Figure \ref{ex:exampleToStopping}). The action of $h_{m}^{\state, \action}$ has also probabilities of $\sfrac{1}{2}$ and leads either to a sink $\zeroSink$ or to $\state'$. 
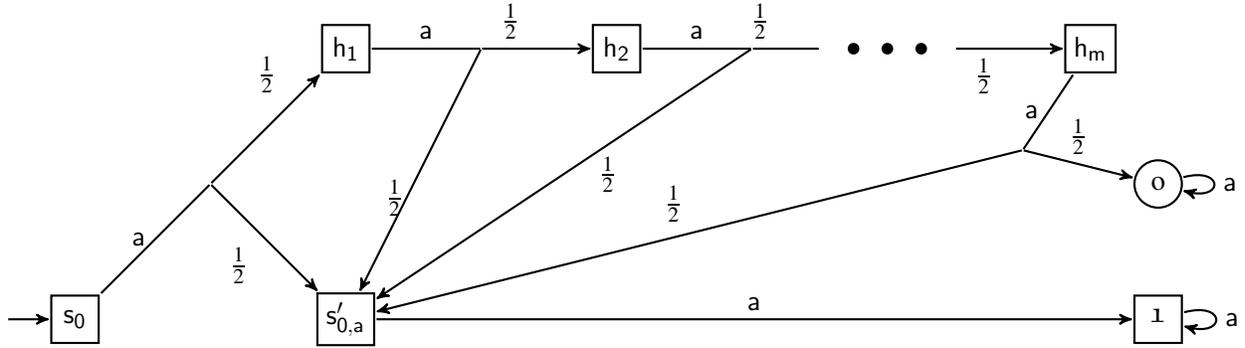
\begin{figure}[htbp]

\centering
\begin{tikzpicture}[thick,scale=0.9, every node/.style={font=\large, transform shape}]
\drawdummy (init) at (-1,0) {};
\drawbox (s0) at (0,0) {$\mathsf{s_0}$};
\drawbox (s0a) at (4,0) {$\mathsf{s_{0,a}'}$};
\drawbox (s1) at (16,0) {$\target$};
\drawdummy (s0Mid) at (2,2) {};

\drawbox (m1) at (4,4) {$\mathsf{h_1}$};
\drawdummy (m1Mid) at (6,4) {};
\drawbox (m2) at (8,4) {$\mathsf{h_2}$};
\drawdummy (m2Mid) at (10,4) {};
\drawdummy (helper1) at (11,4) {};
\drawdummy (helper2) at (13,4) {};
\drawbox (m3) at (15,4) {$\mathsf{h_m}$};
\drawdummy (m3Mid) at (14,2.5) {};
\drawcirc (0) at (16,2)  {$\mathrm{\zeroSink}$};

\filldraw 
(11.5,4) circle (2pt)
(12,4) circle (2pt) 
(12.5,4) circle (2pt);


\draw[->] (init) to (s0);
\draw[-] (s0) to node [left ,midway] {$\mathsf{a}$}(s0Mid);
\draw[->] (s0Mid) to node [text width =1.0cm, align=left, midway, below] {$\frac{1}{2}$} (s0a);
\draw[->] (s0Mid) to node [text width =1.0cm, near end, above] {$\frac{1}{2}$} (m1);
\draw[->] (s0a) to node [above, midway] {$\mathsf{a}$}(s1);

\draw[-] (m1) to node [above,midway] {$\mathsf{a}$}(m1Mid);
\draw[->] (m1Mid) to node [text width =1.0cm, midway, below] {$\frac{1}{2}$} (s0a);
\draw[->] (m1Mid) to node [text width =1.0cm, midway, above] {$\frac{1}{2}$} (m2);

\draw[-] (m2) to node [above,midway] {$\mathsf{a}$}(m2Mid);
\draw[->] (m2Mid) to node [text width =-1.0cm, midway, right] {$\frac{1}{2}$} (s0a);
\draw[-] (m2Mid) to node [text width =1.0cm, midway, above] {$\frac{1}{2}$} (helper1);
\draw[->] (helper2) to node [text width =1.0cm, midway, below] {$\frac{1}{2}$} (m3);

\draw[-] (m3) to node [left,midway] {$\mathsf{a}$}(m3Mid);
\draw[->] (m3Mid) to node [text width =1.0cm, midway, above] {$\frac{1}{2}$} (s0a);
\draw[->] (m3Mid) to node [midway, above] {$\frac{1}{2}$} (0);

\draw[->]  (s1) to [loop right] node [midway,right] {$\mathsf{a}$} (s1);
\draw[->]  (0) to [loop right] node [midway,right] {$\mathsf{a}$} (0);

\end{tikzpicture}
\caption[Example of how to transform into stopping games]{An example of adding the $\epsilon$-transitions to a simple $SG$ where initially there were only $s_0$, $\target$ and the action $\action$ with $\post(s_0, \action) = \post(\target, \action) = \target$. The action of $s_{0,a}$ is simulating the outcome of taking action $\action$ in $s_0$ in the initial $SG$ while $h_1, h_2, ... h_m$ adds the $\epsilon$ needed if the initial $SG$ would be non-stopping}
\label{ex:exampleToStopping}
\end{figure}

\newpage
\section{Proofs for Section \ref{sec:generalQP}} \label{qpProofs}

\subsection{Drop $\noOneSucc$}\label{app:proofNo1Act}
\begin{lemma}[Drop $\noOneSucc$]\label{lem:no1Act}
Let $\state \in \states$ be a state in a $\SG$ with two actions $\action$ and $\actionb$: $\action$ is the action $\state$ originally had, and $\actionb$ is the additional action inserted to comply with $\noOneSucc$.

Then it holds that $\val(\state) = \val(\state,\action)$.
\end{lemma}
\begin{proof}
To prove the lemma, we will make a case distinction based on whether $\state \in \states<\Box>$ or $\state \in \states<\circ>$.

\begin{enumerate}
	\item[case $\state \in \states<\Box>$:] Let the strategy $\strab$ of the Minimizer be arbitrary.
	Let $\straa_b$ be a strategy for the Maximizer in which $\state$ takes action $\actionb$. Per construction, $\actionb$ leads to a sink, and therefore it holds that $\mathsf{\val}_{\straa _b,\strab} (\state) = 0$. Let $\straa_a$ be a strategy for the Maximizer in which $\state$ takes action $\action$. For every state $\state' \in \states$ it holds that $\val(\state') \geq 0$. Thus, $\mathsf{\val}_{\straa _a,\strab} (\state) \geq \mathsf{\val}_{\straa _b,\strab} = 0$.
	It follows that every optimal Maximizer strategy $\straa$ may take action $\action$ in $\state$.
	
	The proof for $\state \in \states<\circ>$ is analogous
	\item[case $\state \in \states<\circ>$:] Let the strategy $\straa$ of the Maximizer be arbitrary.
	Let $\strab_b$ be a strategy for the Minimizer in which $\state$ takes action $\actionb$. Per construction, $\actionb$ leads to a target, and therefore it holds that $\mathsf{\val}_{\straa,\strab_b} (\state) = 1$. Let $\straa_a$ be a strategy for the Minimizer in which $\state$ takes action $\action$. For every state $\state' \in \states$ it holds that $\val(\state') \leq 1$. Thus, $\mathsf{\val}_{\straa,\strab_a} (\state) \leq \mathsf{\val}_{\straa,\strab_b} = 1$.
	It follows that every optimal Minimizer strategy $\strab$ may take action $\action$ in $\state$.
\end{enumerate}
Since $\state$ picks $\action$ in every optimal strategy it holds that $\val(\state) = \val(\state,\action)$.
\end{proof}

\subsection{Drop $\halfProbs$}\label{app:halfProbs}

For the transformation of a $\SG$ into one that complies with $\halfProbs$, we pick state-action pairs ($\state, \action$) with $\state \in \states , \action \in \Av(\state)$ that do not have fitting transition probabilities and construct a corresponding binary tree. Every inner node of the binary tree has only one action with the required transition probabilities. Since we have seen in Lemma \ref{lem:no1Act} that it is not necessary to provide a second action for the inner nodes, we assume they have just one.

We now explain why we do not require the $\SG$ to comply with $\halfProbs$ and that we can use arbitrary rational transition probabilities instead.

Let $\game$ be an SG complying with $\halfProbs$.
Let $v \in \states$ be a state in a binary tree that was constructed to achieve this.
$v$ has only one action $\action$. The value of $v$ is $\val(v) = \sum\limits_{\state' \in \post(v,\action)} \trans(v, \action, \state') \cdot \val(\state')$. 
We construct a modified SG $\game'$ by redirecting all transitions from other states $\state \neq v$ leading to $v$ directly to the successors of $v$ (note that $v$ cannot have a self loop by construction). 
Formally, if $v \in \post(\state,\action)$, then the modified transition function $\trans'$ is defined as follows:
$\trans'(\state,\action,v) = 0$ and for all $\state' \in \post(\state,\action) \setminus \{v\}$ we have
$\trans'(\state,\action,\state ') = \trans(\state,\action ,\state' ) + \trans(\state,\action ,v) \cdot \trans(v,\action ,\state' )$.
Then we get the following chain of equations:
\begin{equation*}
\begin{aligned}
\val<\game>(\state, \action) & = \sum\limits_{\state ' \in \post(\state, \action)} 
\trans(\state ,\action ,\state ') \cdot \val(\state ')\\ 
& = \sum\limits_{\state ' \in (\post(\state, \action) \setminus \{v\})} 
\trans(\state ,\action ,\state ') \cdot \val(\state ') + \delta(\state,\action,v) \cdot \val(v)\\ 
& =
\sum\limits_{\state ' \in (\post(\state, \action) \setminus \{v\})} 
\trans(\state ,\action ,\state ') \cdot \val(\state ') + 
\delta(\state,\action,v) \cdot \sum\limits_{\state ' \in \post(v, \action)} 
\trans(v ,\action ,\state ') \cdot \val(\state ') \\
& = \sum\limits_{\state ' \in \post '(\state, \action)} 
\delta '(\state ,\action ,\state ') \cdot \val(\state ') = \val<\SGinstance '>(\state, \action)
\end{aligned}
\end{equation*}
In the last step, $\post '$ is $\post$ according to $\trans'$.

The value of state $\state$ is no longer dependant on $v$ but only on its successors.
With this construction, we can reverse the $\halfProbs$-transformation into the original $\SG$ by iteratively removing the inner nodes of the trees, as they only have one action with no self looping probability. We preserve the value of every state. \qed

\subsection{Drop Stopping Game}\label{app:proofMEC}

We prove that adding the constraint $\val(\state) = \max\limits_{\straa}\min\limits_{\strab}\valstra(\state)$ for all strategies $(\straa,\strab)$ in the MEC $T$ and all states $\state \in T$ ensures that the resulting QP has a unique solution, namely the correct vector of values.

For this, we proceed in the following steps:
\begin{itemize}
	\item Recall relevant definitions.
	\item Prove that $\valstra(t) = \pr_{t}^{\straa,\strab}(\Diamond \fstates)$ for all $t \in T$.
	\item Follow that $\val(\state) = \max\limits_{\straa}\min\limits_{\strab}\valstra(\state)$ for all strategies $(\straa,\strab)$.
	\item Argue that adding this constraint ensures the convergence of the QP to a unique correct solution.
	\item Show why we can restrict to consider only strategies in the MEC $T$.
\end{itemize}

\begin{enumerate}
	\item Recall the following definitions:
	\begin{itemize}
		\item $\exits=\{(\state,\action) \mid \state \in T \wedge \post(\state,\action) \nsubseteq T \}$ is the set of exiting state-action pairs.
		\item $p^{\straa,\strab}_{e_{(\state,\action)}}(t)$ is the probability to reach the sink state $e_{(\state,\action)}$ in the induced Markov chain $\game[\straa,\strab]$ modified so that the exiting actions lead to sinks.
		\item $\displaystyle \valstra(t) = \sum_{(\state,\action) \in \exits} p^{\straa,\strab}_{e_{(\state,\action)}}(t) \cdot \val'(\state,\action)$ is the value of a state $t \in T$, assuming strategies $\sigma$ and $\tau$ are played.
		
		Note the following important difference between the constraint added to the QP and this proof: The proof uses $\val'(\state,\action)$ instead of $\val(\state,\action)$. 
		The problem is that the symbol $\val$ is overloaded in the context of the QP: it can refer both to the value of the SG as well as to the variable of the QP that eventually converges to the value, but that can have other valuations during the computation.
		
		In the proof, it is important to distinguish the actual value of the SG -- $\val$ -- and the value of the SG assuming we play $(\straa,\strab)$ in the MEC -- $\val'$ -- which is what the variables of the QP are set to.
		However, when adding the constraint to the optimization problem, the distinction between $\val$ and $\val'$ is not necessary, since the valuation of the variable $\val(\state,\action)$ of the QP always depends on the current strategies.
	\end{itemize}

	\item We now prove that $\valstra(t) = \pr_{t}^{\straa,\strab}(\Diamond \fstates)$ for all $t \in T$. In other words, we prove that the computation we use for $\valstra$ actually captures the concept of probability to reach the target under the strategies~$(\straa,\strab)$.
	For this, we use the following chain of equations: 
	\begin{align*}
	\valstra(t) &\eqdef \sum_{(\state,\action) \in \exits} p^{\straa,\strab}_{e_{(\state,\action)}}(t) \cdot \val'(\state,\action)\\
	&= \sum_{(\state,\action) \in \exits} p^{\straa,\strab}_{e_{(\state,\action)}}(t) \cdot \left(\sum_{\state' \in \post(\state,\action)} \trans(\state,\action,\state') \cdot \val'(\state')\right)\\
	&=\sum_{(\state,\action) \in \exits} p^{\straa,\strab}_{e_{(\state,\action)}}(t) \cdot \left(\sum_{\state' \in \post(\state,\action)} \trans(\state,\action,\state') \cdot \pr_{\state'}^{\straa,\strab}(\Diamond \fstates)\right)\\
	&=\pr_{t}^{\straa,\strab}(\Diamond \fstates)
	\end{align*}
	In the last step, we pull together the unfolded probability of the path (going to some exit, taking an exiting action and then continuing from the successor); and we use the fact that all paths reaching the target have to pass through some exiting state-action pair, as otherwise they are stuck in the EC forever.
	
	Note that this relies on the assumption that there is no target state in the MEC; this assumption is justified, since we argued in the preliminaries that every target state has only one action which is a self loop, and we exclude these trivial MECs from consideration, because their value is immediately set correctly to~1.
		
	\item It follows from the previous step that for all $t \in T$ we have 
	\[ \val(t) \eqdef \sup_{\straa} \inf_{\strab} \pr_{t}^{\straa,\strab}(\Diamond \fstates)= \sup_{\straa} \inf_{\strab} \valstra(t) \]
	We overload the symbol $\valstra$ to also denote the probability for states outside the MEC $T$ to reach the targets under strategies $\straa$ and $\strab$, so formally:
	$\valstra(t) \eqdef 
	\begin{cases}
		\sum_{(\state,\action) \in \exits} p^{\straa,\strab}_{e_{(\state,\action)}}(t) \cdot \val'(\state,\action) &\mbox{ if } t \in T\\
		\pr_{t}^{\straa,\strab}(\Diamond \fstates) &\mbox{ otherwise }
	\end{cases}$.
	
	Then, trivially, we also have that for all state $\state \in \states$
		\[ \val(\state) = \sup_{\straa} \inf_{\strab} \valstra(\state) \]
		
	As there are only finitely many memoryless deterministic strategies, and those strategies suffice to attain the optimal value in simple stochastic games, we also have for all state $\state \in \states$
	\[ \val(\state) = \max\limits_{\straa}\min\limits_{\strab}\valstra(\state) \]

	\item 
	The problem of MECs is that in these state sets there are multiple solutions to the Bellman equations (cf. Example \ref{ex:ubNoConvergeVI}), and thus multiple solutions to the quadratic program. By constraining all states in the MECs to $\max\limits_{\straa}\min\limits_{\strab}\valstra(\state)$, we constrain them to exactly their value (by the previous step). 
	Thus, we solve the problem, as now the additional solutions are excluded.
	
	\item 
	Note that so far, this proof considered strategies on the whole state space. 
	However, our algorithm only iterates over all strategies \emph{in the MEC}. 
	
	This is sufficient, because the states outside the MEC are solved by the rest of the QP, and the newly added constraints depend on those solutions, as they depend on the valuation of $\val(\state,\action)$ of exiting state-action pairs.
\end{enumerate}

\qed

\section{Details on the Experiments}
\subsection{Case Studies}\label{app:models}

The case studies coins, prison\_dil, adt, charlton, cdmsn, cloud, mdsm, dice and two\_investors are distributed with PRISM-games 3.0 or available on their case-study-website\\
 \url{http://www.prismmodelchecker.org/games/casestudies.php}.

To judge the impact of a single small MEC, we prepended dice and cdmsn with a single MEC as in Figure \ref{fig:ManyMec}. The exits lead to the initial state of the original model with some probability, and the remaining probability leads to a sink. 

HW and AV are the models used in~\cite{CavLex}; the first two indices show the size of the grid, the last index denotes the single property used.

As interesting handcrafted examples, we used the adversarial model for value iteration from \cite{BVI} (called hm) as well as two newly handcrafted models with either one large MEC (BigMec) or many 3 state MECs (MulMECs).
In BigMec, there is a MEC with two chains of $N$ Maximizer states, see Figure \ref{fig:BigMec}.
In MulMec, a single MEC is repeated $N$ times. For the first $N-1$ repetitions, both exits lead to $\state<0>$ of the next MEC with some probability and to $\state<0>$ of the current MEC with the rest. For the last repetition, both exits lead to some probabilistic combination of target and sink.

\begin{figure}[htbp]

\centering
\begin{tikzpicture}[scale=1, every node/.style={transform shape}]
\drawdummy (init) at (-1,0) {};
\drawcirc (s0) at (0,0) {$\mathsf{s_0}$};
\drawbox (s1) at (1,1) {$\mathsf{s_1}$};
\drawbox (s2) at (1,-1) {$\mathsf{s_2}$};
\drawdummy (mid) at (0.6,-0.2) {};
\drawdummy (upExit) at (2,1) {};
\drawdummy (downExit) at (2,-1) {};

\draw[->] (init) to (s0);
\draw[->]  (s0) to [bend left] (s1);
\draw[-]  (s0) to (mid);
\draw[->]  (mid) to [bend right] (s1);
\draw[->]  (mid) to [bend left] (s2);
\draw[->]  (s1) to [bend left] (s0);
\draw[->]  (s2) to [bend left] (s0);

\draw[->]  (s1) to (upExit);
\draw[->]  (s2) to (downExit);

\end{tikzpicture}
\caption{This is the MEC that is used in the handcrafted MulMec model, as well as in cdmsnMEC and dice50MEC.}
\label{fig:ManyMec}
\end{figure}
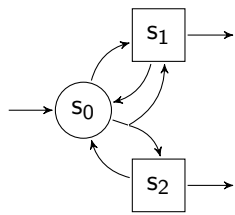

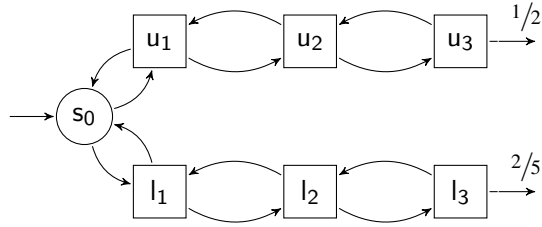
\begin{figure}[htbp]

\centering
\begin{tikzpicture}[scale=1, every node/.style={transform shape}]
\drawdummy (init) at (-1,0) {};
\drawcirc (p) at (0,0) {$\mathsf{s_0}$};
\drawbox (s1) at (1,1) {$\mathsf{u_1}$};
\drawbox (s3) at (1,-1) {$\mathsf{l_1}$};
\drawbox (s2) at (3,1) {$\mathsf{u_2}$};
\drawbox (s4) at (3,-1) {$\mathsf{l_2}$};
\drawbox (s5) at (5,1) {$\mathsf{u_3}$};
\drawbox (s6) at (5,-1) {$\mathsf{l_3}$};
\drawdummy (upExitText) at (5.5,1) {};
\drawdummy (downExitText) at (5.5,-1) {};
\drawdummy (upExit) at (6,1) {};
\drawdummy (downExit) at (6,-1) {};

\draw[->] (init) to (p);
\draw[->]  (p) to [bend right] (s1);
\draw[->]  (p) to [bend right] (s3);
\draw[->]  (s1) to [bend right] (p);
\draw[->]  (s3) to [bend right] (p);

\draw[->]  (s1) to [bend right] (s2);
\draw[->]  (s3) to [bend right] (s4); 

\draw[->]  (s2) to [bend right] (s1);
\draw[->]  (s4) to [bend right] (s3);

\draw[->]  (s2) to [bend right] (s5);
\draw[->]  (s4) to [bend right] (s6);

\draw[->]  (s5) to [bend right] (s2);
\draw[->]  (s6) to [bend right] (s4);

\draw[-]  (s5) to (upExitText);
\draw[-]  (s6) to (downExitText); 

\draw[->]  (upExitText) to node [above, near end] {$\sfrac{1}{2}$}(upExit);
\draw[->]  (downExitText) to node [above, near end] {$\sfrac{2}{5}$}(downExit);

\end{tikzpicture}
\caption{This is our handcrafted scalable model "BigMec" with $N=3$.
	In this model, there are $2 \cdot N +1$ states in the MEC and a dedicated target and sink. 
	The initial state is a Minimizer state, leading to two chains of length $N$ of Maximizer states.
	The fractions on the leaving actions denote the values; they are obtained by a transition to the target state with the given probability and to the sink with the remaining probability.
	The value of every state in the upper chain is $0.5$; the initial state and the lower chain have value $0.4$.}
\label{fig:BigMec}
\end{figure}

\subsection{Detailed Tables of Experimental Results}\label{app:tables}

In this section, we give tables with detailed results for every class of algorithms: Table \ref{tab:vi} for value iteration, Table \ref{tab:si} for strategy iteration and Table \ref{tab:qp} for quadratic/higher order programming.

Every table includes the verification times (in seconds) of several variations of the algorithms, i.e. different optimizations en- or disabled.
An X in the table denotes that the computation did not finish within 15 minutes.
A red background colour indicates that the returned result was wrong, i.e. off by more than the allowed precision. All results that are wrong are off by more than 0.1.

The four left-most columns are shared by all tables. They show the name of the considered case study, its size, the maximum/average number of actions and the number of MECs; for the latter, note that this number excludes trivial MECs, i.e. sink or target states and MECs in regions of the graph that are either not reachable or identified as trivially having value 0/1 by the precomputation.
The case studies are roughly sorted by increasing size/difficulty, with scaled versions of the same model grouped together.

\subsection{Details on the Implementation of Strategy Iteration}\label{app:practicalSI}

In PRISM, even when using strategy iteration to solve the MDP of the opponent, the resulting Markov chain is solved by value iteration.
For all of the models except hm, our solution is still precise, because we detect and fix probabilistic cycles of size 1. 
The model hm is the only one in our benchmark set that has larger probabilistic cycles that cause issues with convergence.

However, this implementation detail does \emph{not} imply that SI is in general not precise. For example, solving the Markov chain with linear programming would result in precise solutions.

\begin{table}[]
	\caption{Table with all experimental results on variations of value iteration. 
	For the general description of the table layout, see Section \ref{app:tables}.
	The considered algorithms are unguaranteed value iteration (VI), the asynchronous simulations-based variant BRTDP~\cite{KKKW18} and bounded value iteration (BVI) as in \cite{KKKW18}.
	As BRTDP is a randomized algorithm, we report both the maximum and median runtime of three tries.
	Considered optimizations for BVI are whether the deflate subroutine is called every loop or only every 100 iterations (1 respectively 100 in the index, suggested in \cite{KKKW18}); and whether to use the topological variant (prepended by T).
	}
\label{tab:vi}
\begin{tabular}{l rrr | rr rrrr}
Case Study       & States & Acts & MECs & VI           & BRTDP & BVI$_1$       & BVI$_{100}$     & TBVI$_1$      & TBVI$_{100}$\\
\toprule

coins            & 19     & 2~$|$~1.16    & 0     & \textless{}1 & \textless{}1         & \textless{}1 & \textless{}1 & \textless{}1 & \textless{}1 \\
prison\_dil      & 102    & 3~$|$~1.34        & 0     & \textless{}1 & \textless{}1         & \textless{}1 & \textless{}1 & \textless{}1 & \textless{}1 \\
adt              & 305    & 4~$|$~1.20        & 0     & \textless{}1 & X                    & \textless{}1 & \textless{}1 & \textless{}1 & \textless{}1 \\
charlton1        & 502    & 3~$|$~1.56        & 0     & \textless{}1 & \textless{}1         & \textless{}1 & \textless{}1 & \textless{}1 & \textless{}1 \\
charlton2        & 502    & 3~$|$~1.56        & 0     & \textless{}1 & \textless{}1         & \textless{}1 & \textless{}1 & \textless{}1 & \textless{}1 \\
cdmsn            & 1,240   & 2~$|$~1.66        & 0     & \textless{}1 & \textless{}1         & \textless{}1 & \textless{}1 & \textless{}1 & \textless{}1 \\
cdmsnMec         & 1,244   & 2~$|$~1.66        & 1     & \textless{}1 & \textless{}1         & \textless{}1 & \textless{}1 & \textless{}1 & \textless{}1 \\
cloud5           & 8,842   & 11~$|$~3.94       & 520   & \textless{}1 & 9 $|$ 6                & \textless{}1 & \textless{}1 & \textless{}1 & \textless{}1 \\
cloud6           & 34,954  & 13~$|$~4.45       & 2176  & \textless{}1 & \textless{}1         & 2            & 2            & 8            & 3            \\
mdsm1            & 62,245  & 2~$|$~1.34        & 0     & 4            & 7 $|$ 7                & 5            & 5            & 3            & 3            \\
mdsm2            & 62,245  & 2~$|$~1.34        & 0     & \textless{}1 & 15 $|$ 14              & \textless{}1 & \textless{}1 & \textless{}1 & \textless{}1 \\
dice20           & 16,915  & 2~$|$~1.45        & 0     & \textless{}1 & 262 $|$ 222            & \textless{}1 & \textless{}1 & \textless{}1 & \textless{}1 \\
dice50           & 96,295  & 2~$|$~1.48        & 0     & 6            & X                    & 6            & 6            & 6            & 6            \\
dice50Mec        & 96,299  & 2~$|$~1.48        & 1     & 6            & X                    & 6            & 6            & 6            & 6            \\
two\_inv   & 172,240 & 3~$|$~1.34        & 0     & 13           & X                    & 14           & 13           & 13           & 13           \\
hw5\_5\_1   & 25,000  & 5~$|$~2.44        & 0     & \textless{}1 & 87 $|$ 18              & \textless{}1 & \textless{}1 & \textless{}1 & \textless{}1 \\
hw5\_5\_2   & 25,000  & 5~$|$~2.44        & 0     & \textless{}1 & \textless{}1         & \textless{}1 & \textless{}1 & \textless{}1 & \textless{}1 \\
hw8\_8\_1   & 163,840 & 5~$|$~2.50        & 0     & 3            & X                    & 3            & 3            & 3            & 3            \\
hw8\_8\_2   & 163,840 & 5~$|$~2.50        & 0     & \textless{}1 & \textless{}1         & \textless{}1 & \textless{}1 & \textless{}1 & \textless{}1 \\
hw10\_10\_1 & 400,000 & 5~$|$~2.52        & 0     & 10           & X                    & 10           & 10           & 10           & 11           \\
hw10\_10\_2 & 400,000 & 5~$|$~2.52        & 0     & \textless{}1 & \textless{}1         & \textless{}1 & \textless{}1 & 1            & 1            \\
AV10\_10\_1      & 106,524 & 6~$|$~2.17        & 0     & \textless{}1 & X                    & \textless{}1 & \textless{}1 & \textless{}1 & \textless{}1 \\
AV10\_10\_2      & 106,524 & 6~$|$~2.17        & 6     & 70           & X                    & 81           & 72           & 82           & 70           \\
AV10\_10\_3      & 106,524 & 6~$|$~2.17        & 1     & 47           & X                    & 46           & 45           & 52           & 55           \\
AV15\_15\_1      & 480,464 & 6~$|$~2.14        & 0     & 1            & X                    & 1            & 1            & 1            & 1            \\
AV15\_15\_2      & 480,464 & 6~$|$~2.14        & 6     & 729          & X                    & X            & X            & X            & X            \\
AV15\_15\_3      & 480,464 & 6~$|$~2.14        & 1     & 492          & X                    & X            & X            & X            & X            \\
hm\_30           & 61     & 1~$|$~1.00        & 0     & \colorbox{red!25}{\textless{}1} & X                    & X            & X            & X            & X            \\
MulMec\_e2    & 302    & 2~$|$~1.99        & 100   & \textless{}1 & X                    & \textless{}1 & 2            & X            & X            \\
MulMec\_e3    & 3,002   & 2~$|$~2.00        & 1000  & 4            & X                    & 36           & 151          & X            & X            \\
MulMec\_e4    & 30,002  & 2~$|$~2.00        & 10000 & \colorbox{red!25}{591}         & X                    & X            & X            & X            & X            \\
BigMec\_e2      & 203    & 2~$|$~1.99        & 1     & \textless{}1 & X                    & \textless{}1 & \textless{}1 & \textless{}1 & \textless{}1 \\
BigMec\_e3      & 2,003   & 2~$|$~2.00        & 1     & \textless{}1 & X                    & 6            & 1            & 6            & 1            \\
BigMec\_e4      & 20,003  & 2~$|$~2.00        & 1     & \colorbox{red!25}{161}          & X                    & 856          & 226          & 871          & 230         \\
\bottomrule
\end{tabular}
\end{table}
\begin{table}[]
		\caption{Table with all experimental results on variations of strategy iteration. 
		For the general description of the table layout, see Section \ref{app:tables}.
		The considered algorithms vary as follows:
		The subscript indicates the method used to solve the MDP, either SI or BVI. The superscript W indicates that the warm start optimization was used (i.e. unguaranteed value iteration to guess a good initial strategy). The prepended T indicates that the topological variant was used.}
	\label{tab:si}
\begin{tabular}{l rrr | rrrr rrrr}
Case Study  & States & Acts    & MECs & SI$_\text{SI}$  & SI$_\text{BV}$       & SI$_\text{SI}^\text{W}$ & SI$_\text{BV}^\text{W}$ & TSI$_\text{SI}$      & TSI$_\text{BV}$      & TSI$_\text{SI}^\text{W}$ & TSI$_\text{BV}^\text{W}$ \\
\toprule
coins       & 19     & 2~$|$~1.16  & 0     & \textless{}1 & \textless{}1 & \textless{}1               & \textless{}1               & \textless{}1 & \textless{}1 & \textless{}1                & \textless{}1                \\
prison\_dil & 102    & 3~$|$~1.34  & 0     & \textless{}1 & \textless{}1 & \textless{}1               & \textless{}1               & \textless{}1 & \textless{}1 & \textless{}1                & \textless{}1                \\
adt         & 305    & 4~$|$~1.20  & 0     & \textless{}1 & \textless{}1 & \textless{}1               & \textless{}1               & \textless{}1 & \textless{}1 & \textless{}1                & \textless{}1                \\
charlton1   & 502    & 3~$|$~1.56  & 0     & \textless{}1 & \textless{}1 & \textless{}1               & \textless{}1               & \textless{}1 & \textless{}1 & \textless{}1                & \textless{}1                \\
charlton2   & 502    & 3~$|$~1.56  & 0     & \textless{}1 & \textless{}1 & \textless{}1               & \textless{}1               & \textless{}1 & \textless{}1 & \textless{}1                & \textless{}1                \\
cdmsn       & 1,240   & 2~$|$~1.66  & 0     & \textless{}1 & \textless{}1 & \textless{}1               & \textless{}1               & \textless{}1 & \textless{}1 & \textless{}1                & \textless{}1                \\
cdmsnMec    & 1,244   & 2~$|$~1.66  & 1     & \textless{}1 & \textless{}1 & \textless{}1               & \textless{}1               & \textless{}1 & \textless{}1 & \textless{}1                & \textless{}1                \\
cloud5      & 8,842   & 11~$|$~3.94 & 520   & \textless{}1 & \textless{}1 & \textless{}1               & \textless{}1               & \textless{}1 & \textless{}1 & \textless{}1                & \textless{}1                \\
cloud6      & 34,954  & 13~$|$~4.45 & 2176  & \textless{}1 & \textless{}1 & \textless{}1               & \textless{}1               & \textless{}1 & \textless{}1 & \textless{}1                & \textless{}1                \\
mdsm1       & 62,245  & 2~$|$~1.34  & 0     & 5            & 5            & 6                          & 6                          & 3            & 3            & 3                           & 3                           \\
mdsm2       & 62,245  & 2~$|$~1.34  & 0     & \textless{}1 & \textless{}1 & \textless{}1               & \textless{}1               & \textless{}1 & \textless{}1 & \textless{}1                & \textless{}1                \\
dice20      & 16,915  & 2~$|$~1.45  & 0     & \textless{}1 & \textless{}1 & \textless{}1               & \textless{}1               & \textless{}1 & \textless{}1 & \textless{}1                & \textless{}1                \\
dice50      & 96,295  & 2~$|$~1.48  & 0     & 7            & 6            & 7                          & 7                          & 6            & 6            & 6                           & 6                           \\
dice50Mec   & 96,299  & 2~$|$~1.48  & 1     & 6            & 7            & 7                          & 7                          & 6            & 6            & 6                           & 6                           \\
two\_inv    & 172,240 & 3~$|$~1.34  & 0     & 38           & 19           & 37                         & 22                         & 11           & 11           & 13                          & 12                          \\
HW5\_5\_1   & 25,000  & 5~$|$~2.44  & 0     & \textless{}1 & \textless{}1 & \textless{}1               & \textless{}1               & \textless{}1 & \textless{}1 & \textless{}1                & \textless{}1                \\
HW5\_5\_2   & 25,000  & 5~$|$~2.44  & 0     & \textless{}1 & \textless{}1 & \textless{}1               & \textless{}1               & \textless{}1 & \textless{}1 & \textless{}1                & \textless{}1                \\
HW8\_8\_1   & 163,840 & 5~$|$~2.50  & 0     & 4            & 4            & 4                          & 4                          & 4            & 4            & 4                           & 4                           \\
HW8\_8\_2   & 163,840 & 5~$|$~2.50  & 0     & \textless{}1 & \textless{}1 & \textless{}1               & \textless{}1               & \textless{}1 & \textless{}1 & \textless{}1                & \textless{}1                \\
HW10\_10\_1 & 400,000 & 5~$|$~2.52  & 0     & 11           & 11           & 11                         & 11                         & 10           & 11           & 11                          & 11                          \\
HW10\_10\_2 & 400,000 & 5~$|$~2.52  & 0     & 2            & 2            & 2                          & 2                          & 1            & 1            & 2                           & 2                           \\
AV10\_10\_1 & 106,524 & 6~$|$~2.17  & 0     & \textless{}1 & \textless{}1 & \textless{}1               & \textless{}1               & \textless{}1 & \textless{}1 & \textless{}1                & \textless{}1                \\
AV10\_10\_2 & 106,524 & 6~$|$~2.17  & 6     & 83           & 79           & 79                         & 80                         & 79           & 76           & 77                          & 79                          \\
AV10\_10\_3 & 106,524 & 6~$|$~2.17  & 1     & 55           & 50           & 58                         & 52                         & 60           & X            & 60                          & X                           \\
AV15\_15\_1 & 480,464 & 6~$|$~2.14  & 0     & 2            & 2            & 2                          & 2                          & 1            & 2            & 2                           & 2                           \\
AV15\_15\_2 & 480,464 & 6~$|$~2.14  & 6     & 797          & 825          & 843                        & 791                        & X            & X            & X                           & X                           \\
AV15\_15\_3 & 480,464 & 6~$|$~2.14  & 1     & 529          & 500          & 512                        & 497                        & X            & X            & X                           & X                           \\
hm\_30      & 61     & 1~$|$~1.00  & 0     & X            & X            & X                          & X                          & X            & X            & X                           & X                           \\
MulMec\_e2  & 302    & 2~$|$~1.99  & 100   & \textless{}1 & X            & \textless{}1               & X                          & \textless{}1 & X            & \textless{}1                & X                           \\
MulMec\_e3  & 3,002   & 2~$|$~2.00  & 1000  & 56           & X            & 56                         & X                          & 3            & X            & 3                           & X                           \\
MulMec\_e4  & 30,002  & 2~$|$~2.00  & 10000 & X            & X            & X                          & X                          & X            & X            & X                           & X                           \\
BigMec\_e2  & 203    & 2~$|$~1.99  & 1     & \textless{}1 & \textless{}1 & \textless{}1               & \textless{}1               & \textless{}1 & \textless{}1 & \textless{}1                & \textless{}1                \\
BigMec\_e3  & 2,003   & 2~$|$~2.00  & 1     & 1            & 1            & 1                          & 1                          & 1            & 1            & 1                           & 1                           \\
BigMec\_e4  & 20,003  & 2~$|$~2.00  & 1     & X            & X            & X                          & X                          & X            & X            & X                           & X     \\
\bottomrule                     
\end{tabular}
\end{table}
\begin{table}[]
	\caption{Table with all experimental results on variations of quadratic/higher order programming. 
		For the general description of the table layout, see Section \ref{app:tables}.
		The considered algorithms include the original method from \cite{DBLP:conf/dimacs/Condon90}, transforming the SG into normal form; and a variant only transforming the SG to a stopping game (\cite{DBLP:conf/dimacs/Condon90}$_\varepsilon$) by adding the $\varepsilon$-transitions if necessary.
		Further, we consider quadratic programming for arbitrary SGs, only transforming the SG to satisfy the $\twoSucc$-constraint (QP), as well as solving a higher-order program for arbitrary SGs (HOP).
		For QP, we differentiate between the used solvers in the subscript (C for CPLEX and G for Gurobi).
		For both QP and HOP, 
		the superscript W indicates that the warm start optimization was used (i.e. unguaranteed value iteration to guess a good initial solution vector), and the prepended T indicates that the topological variant was used.}
	\label{tab:qp}
\begin{tabular}{l rrr| rrrr rrrrr}
Case Study  & States & Acts    & MECs  & \cite{DBLP:conf/dimacs/Condon90} & \cite{DBLP:conf/dimacs/Condon90}$_\varepsilon$      & QP$_\text{C}$        & QP$_\text{G}$        & QP$_\text{G}^\text{W}$ & HOP          & TQP$_\text{G}$       & TQP$_\text{G}^\text{W}$ & THOP         \\ \toprule
coins       & 19     & 2~$|$~1.16  & 0     & X           & X            & \textless{}1 & \textless{}1 & \textless{}1              & \textless{}1 & \textless{}1 & \textless{}1               & \textless{}1 \\
prison\_dil & 102    & 3~$|$~1.34  & 0     & X           & \colorbox{red!25}{8}            & 11           & 9            & 8                         & \textless{}1 & \textless{}1 & \textless{}1               & \textless{}1 \\
adt         & 305    & 4~$|$~1.20  & 0     & X           & \textless{}1 & \textless{}1 & \textless{}1 & \textless{}1              & \textless{}1 & \textless{}1 & \textless{}1               & \textless{}1 \\
charlton1   & 502    & 3~$|$~1.56  & 0     & X           & 179          & X            & 180          & 144                       & \textless{}1 & \textless{}1 & \textless{}1               & \textless{}1 \\
charlton2   & 502    & 3~$|$~1.56  & 0     & X           & X            & X            & X            & X                         & \textless{}1 & X            & X                          & \textless{}1 \\
cdmsn       & 1,240   & 2~$|$~1.66  & 0     & 17          & \textless{}1 & \textless{}1 & \textless{}1 & \textless{}1              & \textless{}1 & \textless{}1 & \textless{}1               & \textless{}1 \\
cdmsnMec    & 1,244   & 2~$|$~1.66  & 1     & X           & \colorbox{red!25}{\textless{}1} & \textless{}1 & \textless{}1 & \textless{}1              & \textless{}1 & \textless{}1 & \textless{}1               & \textless{}1 \\
cloud5      & 8,842   & 11~$|$~3.94 & 520   & X           & X            & X            & X            & X                         & X            & X            & X                          & X            \\
cloud6      & 34,954  & 13~$|$~4.45 & 2176  & X           & X            & X            & X            & X                         & X            & X            & X                          & X            \\
mdsm1       & 62,245  & 2~$|$~1.34  & 0     & X           & X            & X            & X            & X                         & 107          & 6            & 7                          & 3            \\
mdsm2       & 62,245  & 2~$|$~1.34  & 0     & X           & \textless{}1 & 1            & \textless{}1 & \textless{}1              & 5            & \textless{}1 & \textless{}1               & \textless{}1 \\
dice20      & 16,915  & 2~$|$~1.45  & 0     & X           & 3            & X            & 3            & 3                         & 2            & \textless{}1 & \textless{}1               & \textless{}1 \\
dice50      & 96,295  & 2~$|$~1.48  & 0     & X           & X            & X            & X            & X                         & 15           & 6            & 6                          & 6            \\
dice50Mec   & 96,299  & 2~$|$~1.48  & 1     & X           & X            & X            & X            & X                         & 12          & 6            & 7                          & 6            \\
two\_inv    & 172,240 & 3~$|$~1.34  & 0     & X           & X            & \colorbox{red!25}{253}        & X            & X                         & 432          & X            & X                          & 20           \\
HW5\_5\_1   & 25,000  & 5~$|$~2.44  & 0     & 13          & 1            & X            & 1            & 1                         & 3            & \textless{}1 & \textless{}1               & \textless{}1 \\
HW5\_5\_2   & 25,000  & 5~$|$~2.44  & 0     & X           & 1            & X            & 1            & 1                         & 3            & \textless{}1 & \textless{}1               & \textless{}1 \\
HW8\_8\_1   & 163,840 & 5~$|$~2.50  & 0     & 96          & 22           & X            & 20           & 20                        & 17           & 4            & 4                          & 4            \\
HW8\_8\_2   & 163,840 & 5~$|$~2.50  & 0     & X           & 11           & X            & 11           & 11                        & 47           & 2            & 2                          & \textless{}1 \\
HW10\_10\_1 & 400,000 & 5~$|$~2.52  & 0     & X           & 98           & X            & 98           & 98                        & 44           & 12           & 12                         & 11           \\
HW10\_10\_2 & 400,000 & 5~$|$~2.52  & 0     & X           & 48           & X            & 49           & 49                        & 286          & 4            & 4                          & 1            \\
AV10\_10\_1 & 106,524 & 6~$|$~2.17  & 0     & X           & 3            & X            & 3            & 3                         & 9            & \textless{}1 & \textless{}1               & \textless{}1 \\
AV10\_10\_2 & 106,524 & 6~$|$~2.17  & 6     & X           & X            & X            & X            & X                         & X            & X            & X                          & X            \\
AV10\_10\_3 & 106,524 & 6~$|$~2.17  & 1     & X           & X            & X            & X            & X                         & X            & X            & X                          & X            \\
AV15\_15\_1 & 480,464 & 6~$|$~2.14  & 0     & X           & 15           & X            & 10           & 11                        & 41           & 3            & 3                          & 2            \\
AV15\_15\_2 & 480,464 & 6~$|$~2.14  & 6     & X           & X            & X            & X            & X                         & X            & X            & X                          & X            \\
AV15\_15\_3 & 480,464 & 6~$|$~2.14  & 1     & X           & X            & X            & X            & X                         & X            & X            & X                          & X            \\
hm\_30      & 61     & 1~$|$~1.00  & 0     & \colorbox{red!25}{735}         & \textless{}1 & \textless{}1 & \textless{}1 & \textless{}1              & \colorbox{red!25}{\textless{}1} & \textless{}1 & \textless{}1               & \colorbox{red!25}{\textless{}1} \\
MulMec\_e2  & 302    & 2~$|$~1.99  & 100   & X           & X            & X            & \textless{}1 & \textless{}1              & \textless{}1 & \textless{}1 & \textless{}1               & \textless{}1 \\
MulMec\_e3  & 3,002   & 2~$|$~2.00  & 1000  & X           & X            & X            & 4            & 7                         & 3            & 5            & 8                          & 4            \\
MulMec\_e4  & 30,002  & 2~$|$~2.00  & 10000 & X           & X            & X            & X            & X                         & X            & X            & X                          & X            \\
BigMec\_e2  & 203    & 2~$|$~1.99  & 1     & X           & X            & X            & X            & X                         & X            & X            & X                          & X            \\
BigMec\_e3  & 2,003   & 2~$|$~2.00  & 1     & X           & X            & X            & X            & X                         & X            & X            & X                          & X            \\
BigMec\_e4  & 20,003  & 2~$|$~2.00  & 1     & X           & X            & X            & X            & X                         & X            & X            & X                          & X      \\ \bottomrule     
\end{tabular}
\end{table}


}

\end{document}